\documentclass[10pt]{article}

\usepackage[a4paper, margin=0.8in]{geometry} 

\usepackage[utf8]{inputenc} 
\usepackage[T1]{fontenc}    
\usepackage{hyperref}       
\usepackage{url}            
\usepackage{booktabs}       
\usepackage{amsfonts}       
\usepackage{nicefrac}       
\usepackage{microtype}      
\usepackage{graphicx}
\usepackage{amsthm}
\usepackage{amssymb}
\usepackage{cite}
\usepackage{doi}
\usepackage[authoryear,round]{natbib}

\usepackage{subfigure}

\usepackage{array}
\usepackage{tabularx}

\newcolumntype{L}[1]{>{\raggedright\arraybackslash}m{#1}}
\newcolumntype{Y}{>{\raggedright\arraybackslash}X}

\usepackage{enumitem}
\setlist[itemize]{label=$\triangleright$}

\usepackage{wrapstuff}

\usepackage{bm}
\usepackage{amsmath}
\usepackage{algorithm}
\usepackage{algpseudocode}

\DeclareMathOperator*{\argmax}{arg\,max}
\usepackage[table]{xcolor}
\usepackage{changepage}
\usepackage{multirow}
\usepackage{bbm}
\usepackage{todonotes}
\usepackage[most]{tcolorbox}

\theoremstyle{plain}
\newtheorem{thm}{Theorem}[section]
\newtheorem{proposition}[thm]{Proposition}
\newtheorem*{proposition*}{Proposition}
\newtheorem{corollary}[thm]{Corollary}
\newtheorem*{corollary*}{Corollary}

\newtheorem{remark}[thm]{Remark}

\tcolorboxenvironment{proof}{%
  breakable,
  enhanced jigsaw,
  colback=black!5,   
  boxrule=0pt,       
  frame hidden,
  before skip=8pt,
  after skip=8pt,
}

\usepackage{amsthm}
\newtheorem{assumption}{Assumption}

\usepackage{listings}
\usepackage{authblk}

\definecolor{myred}{RGB}{131,0,12}
\title{\textcolor{myred}{\Large{Tempting the Agent: The Economics of Reputation without Persistent Identity in AI Agent Markets}}}

\author[1]{Federico Gatta}
\author[1]{Manuel Naviglio}
\author[2]{Francesco Tarantelli}
\affil[1]{Scuola Normale Superiore, Pisa, Italy}
\affil[2]{Universit\`{a} di Bologna, Italy}

\begin{document}

\maketitle

\begin{abstract} 

Reputation is a fundamental mechanism through which markets sustain trust when service quality cannot be perfectly assessed ex ante, constituting a form of intertemporal economic capital by attracting future demand. Its effectiveness as a disciplinary mechanism depends not only on past interactions but also on the persistence of the identity to which reputation is attached. When identities can be abandoned and recreated cheaply, reputational capital may itself become an object of opportunistic exploitation.

This paper develops a dynamic economic framework to study when reputation is sufficient to discipline autonomous agents. We model reputation as capital attracting future economic activity. At each point, an agent chooses between operating honestly, investing in quality to preserve future gains, or executing a one-shot deviation to extract its reputation's value and restart from a penalized identity. Our analysis relates the temptation to opportunistic behavior to identity-reset costs, reputation persistence, demand sensitivity, and enforcement design, deriving comparative statics on optimal quality provision.

Autonomous AI-agent operating on the blockchain are a relevant application: infrastructures such as \texttt{ERC-8004}, \texttt{ERC-8183}, and \texttt{x402} combine reputation, identity, and payments in permissionless markets. Nonetheless, our framework applies to any environment where reputation generates future business and identities are replaceable.

\end{abstract}



\section{Introduction}

Reputation is a fundamental mechanism through which economic relationships can be sustained when contracts are incomplete and the quality of a service cannot be perfectly assessed ex ante. A seller, professional, platform participant, or service provider that repeatedly interacts with customers accumulates a history of past behavior. When future users condition their decisions on this history, reputation becomes economically valuable: a better and more established reputation can attract more demand, support larger transaction volumes, and generate higher future revenues. In this sense, reputation is not merely an informational score. It is a form of \emph{reputational capital} whose value derives from the future economic activity that it allows its holder to attract.

The institutional structure surrounding reputation is equally important. In many centralized environments, accounts are embedded in systems in which identities are relatively persistent. Platform operators can suspend or exclude participants, retain transaction histories, require verified payment instruments or business credentials, and, in some cases, rely on contractual relationships or Know-Your-Customer (KYC) procedures. Even when perfect real-world identification is absent, these mechanisms can create a degree of \emph{economic identity persistence}: misconduct may lead not only to a lower rating, but also to exclusion, legal liability, loss of access, or substantial costs of re-entry.

The situation changes substantially when the economic actor behind a service is not persistently identifiable or economically punishable. Consider an environment in which an autonomous software agent provides services, receives payments, accumulates ratings, and interacts directly with humans or with other autonomous agents. If the agent is not associated with a verified and persistent real-world principal, its observable identity may be easier to abandon and recreate. Reputation remains valuable because it attracts future business, but the agent or the entity controlling it has the option of a one-shot deviation: destroy the existing reputational capital, appropriate a sufficiently large short-run payoff, and subsequently re-enter the market under a new or partially reset identity.

This possibility creates a distinctive dynamic incentive problem. An agent with a valuable reputation faces a trade-off between the continuation value generated by preserving that reputation and the immediate value that can be extracted through opportunistic behavior. Remaining honest generates a stream of future revenues, but may require costly effort to deliver the promised service and preserve the quality of the reputation record. Deviating may instead allow the agent to appropriate some of the economic value currently entrusted to it, while sacrificing its existing reputation and incurring the cost of rebuilding a new identity. The relevant comparison is therefore between the value of preserving reputational capital and the value of opportunistically liquidating it. When identity reset is sufficiently cheap, reputational losses need not be permanent, and reputation alone may fail to discipline behavior.

The absence of a persistent one-to-one mapping between observable identities and economically accountable principals creates a second, related concern. If identities can be produced at low cost, the same economic entity may be able to operate multiple agents simultaneously. These identities can potentially be used not only to replace an agent after misconduct, but also to manipulate ratings, generate artificial activity, coordinate feedback, or otherwise distort the informational content of the reputation system. This is the classical connection between cheap pseudonyms, whitewashing, and Sybil behavior, but it acquires renewed importance in autonomous-agent economies, where the creation, operation, and coordination of multiple software identities can itself be automated. The present paper focuses on the first mechanism---the dynamic incentive to preserve or abandon an existing reputation---while cheap identity creation and reputation manipulation provide a broader motivation and a natural extension of the framework.

These issues become particularly relevant as autonomous Artificial Intelligence (AI) agents begin to act as economic participants rather than merely as passive software tools. An AI agent may provide a digital service to a human user, but it may also contract another agent to perform a specialized task, purchase access to data or computational resources, delegate part of a workflow, or receive and manage funds on behalf of another economic actor. In such an environment, humans may transact with agents, agents may transact with other agents, and autonomous entities may act as buyers, sellers, intermediaries, evaluators, and counterparties.

Unlike traditional firms or professionals, however, the observable identity through which an autonomous agent interacts with the market need not be persistently tied to an economically accountable human or organization. While a conventional market participant may face legal liability, contractual sanctions, platform exclusion, or long-lasting reputational consequences following misconduct, an autonomous agent operating under a permissionless identity may instead be abandoned and replaced at relatively low cost, leaving only limited consequences for the underlying controller. As a result, many of the disciplinary mechanisms that sustain trust in conventional markets may become substantially weaker.

This paper studies the economic consequences of this form of identity structure. Our central question is:

\begin{quote}
\emph{When can reputation alone discipline an economic agent whose observable identity is not persistently tied to an economically accountable principal, and can therefore be abandoned and recreated without imposing substantial consequences on its underlying controller?\\ In other words: is reputation something that the agent considers as a capital that fears to lose, and is this enough to incentivize it to behave well?}
\end{quote}

We develop a dynamic microeconomic model in which an agent accumulates reputation through repeated interaction and earns economic returns from the demand generated by that reputation. The state of the agent is summarized by two variables. The first captures the quality of its reputation, while the second captures the depth or maturity of the underlying track record. This distinction reflects the fact that an identical reputation score may have very different economic value depending on whether it is supported by only a few interactions or by a long and established history.

Reputation affects the volume of economic activity entrusted to the agent. Better and more mature reputations, therefore, generate larger business opportunities and higher future income. Conditional on continuing to operate, the agent chooses the quality or effort devoted to service provision: greater effort is costly but improves future reputation. At the same time, the agent can execute a one-shot opportunistic deviation. In that case, it extracts a fraction of the economic value currently associated with the activity it intermediates, destroys or abandons its existing reputational position, and subsequently restarts from a penalized state. The framework, therefore, treats reputation as an intertemporal productive asset and opportunistic deviation as a possible liquidation of that asset.

This formulation makes explicit the role played by identity persistence. If re-entry is costly and rebuilding reputation is slow, destroying a mature identity implies a large loss of future income and reputation can sustain honest behavior even in the absence of direct legal enforcement. If, instead, a new identity can be created cheaply and a credible reputation can be reconstructed rapidly, the disciplinary effect of reputation becomes weaker. The model characterizes the regions of the state space in which continuation or deviation is optimal and studies how these regions depend on the value extractable from misconduct, the cost of identity reset, the persistence of reputation, the sensitivity of demand to reputation, and the agent's discount factor.

We additionally introduce staking and slashing as a complementary economic enforcement mechanism. A stake tied to the activity intermediated by the agent increases the amount that can be lost following misconduct and therefore alters the comparison between opportunistic extraction and the continuation value of reputation. The framework also highlights a potential trade-off: stronger staking may reduce the incentive to deviate while simultaneously increasing the cost of honest operation. Enforcement can therefore affect not only whether the agent behaves honestly, but also the effort optimally supplied when it remains in the market.

The theoretical mechanism is deliberately more general than any particular technological implementation. It applies whenever three conditions are present: reputation affects future demand, opportunistic behavior can generate a sufficiently large one-shot gain, and economic identity can be easily abandoned and is cheap to recreate. These conditions can arise in conventional digital markets, pseudonymous online services, decentralized marketplaces, and autonomous-agent economies.

Blockchain-based systems provide an especially transparent application of
this problem because they combine programmable payments, publicly observable
transaction histories, and permissionless identity creation. A first step
toward enabling autonomous economic interactions in this environment came
with the launch of the x402 protocol in May 2025, which introduced an
HTTP-native payment mechanism allowing software clients, including autonomous
agents, to purchase APIs, data, computation, and other digital resources
through programmatic stablecoin payments \citep{x402docs2026}. Although x402 operates at the application and HTTP layers, the underlying
payments can settle on-chain. It therefore represented an early market-level
implementation of autonomous machine-to-machine commerce, predating Ethereum
standards specifically designed to provide agents with identity, reputation,
and job-escrow primitives.

Ethereum standardization followed this application-layer experimentation.
ERC-8004, proposed in August 2025, introduced registries for agent identity,
reputation, and validation \citep{derossi2025erc8004}. The proposal explicitly
treats payments as orthogonal to its trust layer and discusses x402 payments
as a complementary source of economically meaningful interaction data.
ERC-8183, proposed in February 2026, subsequently introduced a standardized
on-chain job-escrow primitive in which clients fund tasks, providers submit
work, and evaluators determine whether escrowed payments are released or
refunded \citep{crapis2026erc8183}. Its specification also explicitly
contemplates compatibility with HTTP-native payment protocols such as x402.

This chronology is economically informative. Rather than originating all
components of the agent economy within the Ethereum standards process,
payment and autonomous-commerce primitives first developed at the
application layer and were subsequently complemented by ERC proposals that
formalized identity, reputation, validation, and escrow on-chain. The
evolution from x402 to ERC-8004 and ERC-8183 can therefore be interpreted as
a bottom-up standardization process in which Ethereum's standards layer
codifies and makes composable primitives that had already begun to emerge in
the broader agentic-commerce ecosystem.

This setting makes the theoretical trade-off particularly concrete. An \texttt{ERC-8004} identity may accumulate a valuable public history, while \texttt{ERC-8183} or \texttt{x402}-like payment mechanisms may allow that reputation to translate into increasingly valuable task flow or payment activity. At the same time, permissionless identity creation means that abandoning one identity and creating another may be significantly easier than replacing a legally identified economic actor in a conventional market. The blockchain application, therefore, does not define the theory developed in this paper; rather, it provides a particularly clean institutional environment in which the underlying problem of reputation without persistent identity becomes observable.

The paper is organized as follows. The remainder of this section discusses the related literature. Section \ref{sec:institutional} introduces the concepts of reputational capital and identity persistence, highlighting how they differ in the context of autonomous agents. Section \ref{sec:model} presents the model describing the dynamics of the system, introducing the state variables, control variables, and economic flows. Section \ref{sec:dynamic} formulates the agent's dynamic optimization problem and studies how reputation simultaneously creates discipline and temptation. Section \ref{sec:sur} specializes the analysis to the small-update regime, where the model admits an explicit analytical characterization, and we are able to explicitly extract the effect that staking could have on agent's behavior. Section \ref{sec:numerical} provides some numerical results corroborating the theory previously developed. Section~\ref{sec:effort} shifts the focus to the agent's optimal quality choice conditional on honest continuation. Section \ref{sec:conc} concludes. The Appendix is organized as follows: Appendix \ref{app_symbols} summarizes the main symbols used in the paper; Appendix \ref{app_proofs} contains the proofs of the results stated in the main text.

\subsection{Related Literature}

Our contribution is related to several strands of literature. Classical models of reputation study how repeated interaction and future rents can sustain quality provision \citep{klein1981role,shapiro1983premiums,honer2002reputation}. Work on online marketplaces shows how reputation systems affect seller behavior and how feedback mechanisms can themselves be shaped by reciprocity, timing, and strategic interaction \citep{resnick2002trust,dellarocas2003digitization,cabral2010dynamics,bolton2013engineering}. Particularly close to our setting is the literature on cheap pseudonyms and whitewashing, which emphasizes that low-cost identity replacement weakens the disciplinary role of reputation \citep{friedman2001social,bar2008seller}. We build on this insight but study it in an environment in which reputation determines the scale of economic activity entrusted to an autonomous agent, creating an endogenous tension between the future value of reputational capital and the value that can be extracted from that activity through a one-shot deviation.

Our paper is also related to the emerging literature on trust and evaluation in autonomous-agent systems. Agent evaluation is intrinsically multidimensional, and the quality of an agent cannot generally be summarized by a single measure of predictive accuracy or user satisfaction \citep{kapoor2024agents}. Dynamic approaches to agent ranking similarly emphasize that credibility depends on observed interactions, performance, and the information accumulated over time \citep{omidshafiei2019alpharank,lou2025drf}. Recent work on inter-agent trust further stresses that reputation is only one component of trustworthy agent interaction and may need to be complemented by proof, stake, validation, or other constraints \citep{hu2025interagent}. Our analysis provides an economic foundation for this argument by identifying conditions under which reputation is insufficient on its own and by showing how financial penalties can alter the agent's incentives.

Finally, blockchain and decentralized systems provide a natural environment for studying the interaction between pseudonymity and economic incentives. Previous work has examined decentralized reputation systems, privacy-preserving trust mechanisms, and blockchain-based trust management \citep{bellini2020blockchain,pal2021blockchain,schaub2016trustless,dimitriou2021decentralized,hasan2022privacy}. Other studies document opportunistic behavior in decentralized markets, including scam tokens, rug pulls, sniper bots, and MEV extraction \citep{xia2021trade,cernera2022tokens,weintraub2022flashbot,tran2025cook,gogol2025firstspammed,naviglio2026price,di2026deviations}. Our focus is complementary: rather than identifying a particular class of attacks, we characterize the dynamic incentive problem that arises when an economically active entity can accumulate valuable reputation while retaining the option to abandon its identity after misconduct.

\section{Reputation and Identity in Autonomous Agent Economies}
\label{sec:institutional}
This section introduces the concepts of reputational capital and discusses some issues related to the presence of autonomous agents as market participants. Finally, it discusses how the presence of autonomous agents is currently handled by existing blockchain protocols.

\subsection{Reputational Capital}

The starting point of our framework is that reputation affects future economic activity. A service provider with a stronger and more established reputation is more likely to be selected by counterparties, receive a larger number of requests, be assigned more valuable tasks, or intermediate greater economic flows. Reputation therefore generates future rents and can be interpreted as a stock of \emph{reputational capital}.

We distinguish between two dimensions of this capital. The first is the \emph{reputation level}, which summarizes how positively an entity has been evaluated. The second is the \emph{reputation maturity}, which captures the amount of observable evidence supporting that evaluation. This distinction is economically relevant because an identical reputation score may have very different implications depending on whether it is based on only a few observations or on a long and established interaction history.

We denote these two dimensions by $F$ and $M$, respectively. Their joint economic effect is summarized by a demand (or activity) function $v(F,M)$, which maps reputational capital into the volume of economically meaningful activity entrusted to the agent. We assume that both a higher reputation level and a more mature reputation record attract weakly more activity, so that $v(F,M)$ is increasing in both arguments.

The key mechanism is therefore:
\[
(F,M)
\quad \longrightarrow \quad
v(F,M)
\quad \longrightarrow \quad
\text{future economic value}.
\]

This relationship is deliberately general. The model does not require reputation to be summarized by a universal score, nor does it require a specific feedback architecture. What matters is that counterparties use observable information about past behavior to condition future economic interaction.

\subsection{Identity Persistence and Reset}

The disciplinary power of reputation depends not only on how reputation affects future demand, but also on how costly it is to lose the identity to which that reputation is attached. When identity is highly persistent and tied to an economically accountable principal, misconduct can generate consequences that extend beyond the loss of reputation itself. By contrast, when an agent can continue operating after misconduct without exposing an identifiable underlying controller to substantial legal, contractual, or economic consequences, reputational deterioration may constitute the main disciplinary cost generated by deviation.

We represent the consequences of misconduct through a \emph{post-deviation penalized state}. Following opportunistic behavior, the agent loses part or all of the reputational capital associated with its current position and moves to a lower-value state. This transition may correspond to different underlying situations. The agent may abandon its current identity and re-enter under a new one, or it may continue operating under the same identity with a substantially deteriorated reputation and a thinner effective track record. What matters for the model is not the specific implementation of the transition, but the fact that misconduct moves the agent to a penalized state.

The economic cost of this transition has two components. First, the agent loses the future economic value associated with the reputational capital destroyed by the deviation. Second, it incurs a direct cost $\kappa^R$ associated with restoring economically meaningful activity after misconduct. Depending on the environment, this cost may include the technical cost of establishing a new identity, the resources required to rebuild reputation, the loss of access to counterparties, or other frictions associated with returning to a valuable market position.

The relevant friction is therefore broader than the cost of identity creation itself. Even if the agent does not abandon its identity, or if creating a replacement identity is technically inexpensive, recovering from a sufficiently penalized reputational state may remain economically costly because demand must be rebuilt over time. This distinction is central to the incentive problem studied below. Reputational discipline is effective only if the value destroyed by misconduct, together with any direct post-deviation cost or financial penalty, is sufficiently large relative to the immediate gain from opportunistic behavior. When the agent can rapidly return to economically meaningful activity from the penalized state, and misconduct does not generate substantial consequences for an identifiable underlying principal, the deterrent effect of reputation is correspondingly weaker.

\subsection{Blockchain-Based Autonomous Agent Markets}
\label{sec:blockchain_application}

Permissionless blockchain systems provide a particularly transparent implementation of this environment. They combine publicly observable interaction histories with cryptographic identities that are not necessarily associated with a verified real-world principal. At the same time, new addresses or protocol-level identities can often be created at relatively low technical cost. Public transparency, therefore, does not automatically imply persistent accountability. A prominent example is \texttt{ERC-8004}, which introduces a modular infrastructure for autonomous agents based on three registries: an Identity Registry, a Reputation Registry, and a Validation Registry \citep{derossi2025erc8004}.

\paragraph{Identity Registry.}
The Identity Registry represents each agent through an ERC-721 token associated with an agent card. The token provides a persistent protocol-level identifier and records the blockchain address controlling the agent. The associated metadata can describe capabilities, service endpoints, and other agent-specific information. For the purposes of our framework, the relevant distinction is between \emph{protocol identity} and \emph{economic identity}: on-chain, \texttt{ERC-8004} provides only the former.

\paragraph{Reputation Registry.}
The Reputation Registry allows clients to attach publicly observable feedback to an agent identity. Feedback may include numerical evaluations as well as additional semantic information, while the standard does not impose a unique aggregation rule. In our model, the reputation state $F$ should not be interpreted as a protocol-level universal score. Instead, it represents the normalized reputation measure used by the relevant marketplace, application, or counterparty. Different environments may aggregate the same public feedback in different ways. The maturity state $M$ captures the amount or depth of observable evidence supporting that reputation.

\paragraph{Validation Registry.}
The Validation Registry allows additional attestations or proofs to be associated with an agent. These may include cryptographic proofs, hardware-backed attestations, or other mechanisms intended to strengthen trust in the agent's behavior or execution environment. From the perspective of our model, validation is complementary to reputation. A validation mechanism can reduce the expected gain from misconduct by making deviation harder, increasing the likelihood that misconduct is detected, or enabling financial penalties. We mathematically discuss the impact of one of such mechanisms: \emph{Staking and Slashing} (S\&S). This requires the agent to stake an amount proportional to the handled volume: in case of misbehavior, the stake is slashed.

\subsection{Agentic Commerce and Economic Flow}

Other protocols are gaining popularity in agentic commerce. \texttt{ERC-8183} defines a job-escrow mechanism in which a client funds a task, a provider performs the required service, and an evaluator determines whether the task is completed or rejected, resulting in payment release or refund \citep{crapis2026erc8183}. This creates a direct link between an agent's ability to attract tasks and the amount of economic value entrusted to it. The \texttt{x402} protocol provides a complementary mechanism for programmatic payments over HTTP. It allows software clients, including autonomous agents, to purchase access to APIs, data, computation, or other digital services without relying on a traditional account-based payment flow \citep{x402docs2026}. Recent work on agentic commerce emphasizes that these architectures introduce additional issues involving transaction authorization, payment-service atomicity, privacy, evaluator integrity, and inter-agent trust \citep{mao2026sok,li2026fiveattacks,li2026a402,stantchev2026hardeningx402}.

These protocols provide a concrete interpretation for the economic flow $v(F,M)$ used in the model. A better and more mature reputation may allow an agent to receive more requests, larger tasks, or greater payment flows. Under honest operation, this activity generates fee income. Under opportunistic behavior, the same flow may become the object of economic extraction. The key point is therefore that reputation has a dual effect. It increases the continuation value of remaining reputable, because it generates future business, but it may also increase the immediate amount of value available for opportunistic extraction. A highly reputable agent can therefore become both more valuable to preserve and more tempting to exploit.

\texttt{ERC-8004}, \texttt{ERC-8183}, and \texttt{x402} jointly provide a concrete example of this mechanism. \texttt{ERC-8004} supplies persistent protocol-level identity and reputation signals, while \texttt{ERC-8183} and \texttt{x402} provide channels through which economically meaningful task and payment flows can be generated. At the same time, the underlying infrastructure remains permissionless and does not necessarily enforce a persistent mapping between each agent identity and a verified real-world principal. The theoretical framework we develop below abstracts from these implementation details and studies the underlying economic mechanism directly.

%

\section{Model Environment}
\label{sec:model}

We now formalize the economic mechanism described in the previous sections. We consider a forward-looking economic agent repeatedly providing a service to a sequence of counterparties. Our model does not rely on a specific technological infrastructure. The agent accumulates reputational capital through repeated interaction. A stronger and more established reputation attracts a larger volume of economically meaningful activity and therefore generates future income. At the same time, the agent may sacrifice its current reputational position in exchange for an immediate opportunistic payoff and subsequently re-enter the market from a penalized state. The central economic trade-off is therefore between preserving reputational capital and extracting the value currently associated with it.

\subsection{State Variables and Control}

Let us consider a discrete-time model with equally spaced time increments and an uninterrupted flow of events. In the following, we denote the state space with $\mathcal{X} := [0,1]^2$. We denote the component-wise order on $\mathcal{X}$ with $\preceq$. That is, given $(F_1, M_1), (F_2, M_2) \in \mathcal{X}$, it holds $(F_1, M_1) \preceq (F_2, M_2)$ if and only if $F_1 \le F_2$ and $M_1 \le M_2$. We model the state of the system as a pair $\{(F_t, M_t)\}_{t\in\mathbb{N}}$. $F_t \in [0,1]$ denote the agent's \emph{average reputation score} at time $t$, and $M_t \in [0,1]$ denotes the \emph{maturity} or \emph{depth} of the agent's observable reputation record. Both are assumed to be normalized quantities within reliable boundaries.

The state variable $F_t$ summarizes the effective reputation measure relevant for economic decisions in the environment under consideration. In addition, $M_t$ captures how much publicly observable evidence supports that reputation.
We define $M_t$ as a normalized, exponentially weighted measure of the economically meaningful interaction evidence supporting the agent current reputation. Specifically, starting from $M_0 \in [0,1]$, $M_t$ is recursively defined as:
\begin{equation}
    M_{t+1} = (1-\rho) \, M_t+\rho\,\sigma(F_t,M_t), \qquad \rho \in [0,1], \qquad \sigma:\mathcal{X}\rightarrow[0,1]
\end{equation}
$\sigma$ can be thought of as the current normalized interaction signal, which enters with weight $\rho$; signals from earlier periods receive geometrically declining weights. Thus, $M_t$ captures the recency-weighted depth of the agent track record rather than its cumulative handled volume. Intuitively, a high reputation score supported by a very thin history of observations need not generate the same demand, nor the same incentive effects, as the same score supported by a long and well-established record.\\

The two state variables, therefore, play distinct but complementary roles. The variable $F_t$ captures \emph{how positively} the agent is evaluated on average, while $M_t$ captures \emph{how solidly established} that evaluation is. This distinction will be crucial below, because opportunistic incentives depend not only on the score itself, but also on how much demand that score is able to support in virtue of the depth of the underlying record.

Importantly, the framework does not impose a unique interpretation or aggregation rule for reputation. Different markets, platforms, or applications may process observable evaluations differently and assign different weights to the available information. The explicit dynamical specification introduced below allows us to study the incentive problem without committing to a particular technological implementation.

As discussed above, the central incentive problem studied in this paper is whether the agent finds it optimal to continue operating honestly or to opportunistically misbehave to exploit the entrusted volume related to its current reputation. At each period, the agent first chooses between two actions. It may continue operating honestly or execute a one-shot opportunistic deviation. We denote this binary decision by $h_t\in\{0,1\}$, where $h_t=1$ corresponds to honest continuation and $h_t=0$ to opportunistic deviation. In decentralized applications, such a deviation may take the form of a \emph{rug pull}.

When the agent chooses to deviate, it appropriates the economic value currently entrusted to it and subsequently damages its current reputation. The damage could result in a lower reputation value, caused by the negative feedback, or in the creation of a new identity. The latter may occur because misconduct becomes publicly observable or because the agent deliberately abandons its current identity and restarts from a penalized state.

Conditional on choosing honest continuation ($h_t=1$), the agent also selects a service-quality level $q_t\in[0,1]$. The variable $q_t$ represents the effort, computational resources, engineering practices, or technological standards devoted to service provision. Higher quality is more costly but better improves future reputation. The agent's control at time $t$ is therefore the pair $(h_t,q_t)$, although $q_t$ is relevant only when $h_t=1$.

\subsection{State Transition and Reputation Dynamics}
\label{sec:transition}

We now specify how the reputational state evolves following either honest continuation or opportunistic deviation. We adopt a deterministic mean-field approximation, so that the state at time $t+1$ is represented as a deterministic function of the current state and of the agent's action. This amounts to replacing the realized sequence of individual feedback events and interactions with their expected effect on the aggregate reputational state.

When the agent chooses honest continuation, $h_t=1$, it provides service quality $q_t$ and subsequently receives feedback. Let $s_t=s(q_t)\in[0,1]$ denote the feedback generated by quality $q_t$, where $s:[0,1]\rightarrow[0,1]$ maps service quality into the corresponding evaluation.

Contrastingly, the maturity of the reputation record evolves, by definition, as a function of the amount of observable interaction generated by the current state: $\sigma(F,M)$.
The honest transition (originating from the honest behavior, $h_t=1$) is therefore described by the following assumption.

\begin{assumption}[Honest transition]
\label{ass:update}
The honest transition $T:\mathcal{X}\times[0,1]\rightarrow\mathcal{X}$ is defined as
\[
T(F,M,q) := \big(T_F(F,q), T_M(F,M)\big)
\]
where
\[
T_F(F,q) := (1-\lambda) \, F+\lambda \, s(q), \qquad\qquad T_M(F,M) = (1-\rho) \, M+\rho\,\sigma(F,M),
\]
with $\lambda,\rho\in[0,1)$. The feedback function $s$ satisfies
\[
s'(q)>0, \qquad\qquad s''(q)\leq0.
\]

The interaction function $\sigma$ satisfies
\[
\sigma(F,0)=0,
\]
together with
\[
\partial_F\sigma(F,M)\geq0, \qquad\qquad \partial_M\sigma(F,M)\geq0,
\]
\[
\partial^2_{F,F}\sigma(F,M)\leq0, \qquad\qquad \partial^2_{M,M}\sigma(F,M)\leq0,
\]
and
\[
\partial^2_{F,M}\sigma(F,M)\geq0.
\]
\end{assumption}

The parameters $\lambda$ and $\rho$ govern the speed at which reputation and
maturity incorporate new information. Larger values imply that recent
outcomes have a stronger impact on the state, whereas smaller values generate
greater persistence. Thus, when $\lambda$ and $\rho$ are small, reputational
capital behaves as a slowly evolving stock accumulated through repeated
interaction. The assumption $s'(q)>0$ captures the fact that higher service quality
generates better feedback, while $s''(q)\leq0$ allows for diminishing returns
to quality. Similarly, the assumptions on $\sigma$ imply that stronger
reputation and a deeper existing track record generate more observable
activity, while the corresponding marginal effects may decline as the state
becomes more mature.\\

The consequences of opportunistic deviation need not coincide with the honest
transition. We therefore introduce a separate post-deviation transition:
\begin{assumption}[Post-deviation transition]
\label{ass:deviation_transition}
The post-deviation transition
\[
R:\mathcal{X}\rightarrow\mathcal{X}, \qquad\qquad R(F,M) := \big( R_F(F,M), R_M(F,M) \big),
\]
is continuous and order-preserving with respect to the componentwise order on
$\mathcal{X}$. That is, for any
$(F_1,M_1),(F_2,M_2)\in\mathcal{X}$ satisfying $(F_1,M_1) \preceq (F_2,M_2)$, we have $R(F_1,M_1) \preceq R(F_2,M_2)$
Moreover, $R$ is independent of the staking parameter $\kappa^S$.
\end{assumption}

This formulation allows the model to encompass different responses to misconduct. One possibility is that the agent remains active under the same identity after the deviation. In this case, misconduct may generate an adverse feedback signal and therefore reduce the reputation score without mechanically destroying the previously accumulated track record. For example, if deviation generates a fixed feedback score $s^{dev}$, with $s^{dev}$ close to zero, a natural specification is
\begin{equation}
\label{eq:R_same_identity}
R^{same}(F,M) := \left((1-\lambda) \, F+\lambda \, s^{dev}, \, (1-\rho) \, M+\rho\,\sigma(F,M) \right).
\end{equation}
The reputational punishment then operates mainly through the deterioration of $F$, which reduces subsequent economic activity through the volume function $v(F,M)$, while the maturity accumulated under the existing identity is preserved and may continue to evolve.

A second possibility is that sufficiently severe or verifiable misconduct forces the agent to abandon the compromised identity and re-enter the market under a new one. In this case, both the reputation level and the maturity of the public record are reset:
\begin{equation}
\label{eq:R_new_identity}
R^{new}(F,M) := \big( F_{\mathrm{reset}}, M_{\mathrm{reset}} \big).
\end{equation}
A fresh identity naturally corresponds to low values of $F_{\mathrm{reset}}$ and $M_{\mathrm{reset}}$. In particular, under the baseline specification below we allow for a small positive $M_{\mathrm{reset}}$, representing the minimal initial exposure or activity required for a new identity to enter the market. This avoids making the zero-maturity state mechanically absorbing. The agent then loses not only its previous reputation score but also the accumulated depth of the track record supporting that reputation.

The second mechanism therefore imposes a stronger form of reputational discipline. Under continued operation $R^{same}$, the agent retains part of the reputational capital accumulated before misconduct, particularly through the maturity dimension. Under identity replacement $R^{new}$, by contrast, the entire economic value associated with the old reputational state has to be rebuilt. This distinction highlights the potential importance of explicit signaling, verification, and enforcement mechanisms capable of associating misconduct with the corresponding market identity and, when appropriate, preventing the compromised identity from simply preserving its accumulated reputational position.

Rather than imposing either mechanism at the outset, the theoretical analysis below is formulated in terms of the general post-deviation map $R$. The two cases in Eqs.~\eqref{eq:R_same_identity} and \eqref{eq:R_new_identity} will then be studied as economically meaningful special cases and compared to stress the importance of explicit signaling.

\subsection{Economic Flow and Payoffs}

We model the agent's economic flow $v_t$, which is the volume of economically meaningful activity intermediated by the agent, as a function $v(F_t, M_t):\mathcal{X} \rightarrow \mathbb{R}_+$ of both reputation and maturity.

\begin{assumption} \label{ass:volume}
The volume function $v$ satisfies the following:
\[
    v(0,M) = v(F,0) = 0 \quad  \partial_F v(F,M), \; \partial_M v(F,M) \ge 0, \quad \partial^2_{F,F} v(F,M), \; \partial^2_{M,M} v(F,M) \le 0, \quad \partial^2_{F,M} v(F,M)\ge 0
\]
\end{assumption}

The function $v$ captures the reduced-form response of the market to the publicly observable reputation signal. $v(0,M) = v(F,0) = 0$ models the fact that reputation without maturity is worthless, as well as a mature, bad reputation. The assumption on the first derivatives describes the fact that both a higher average reputation and a more mature reputation record increase expected activity. The assumption on the second derivatives describes the concavity of $v$ when viewed as a function of only the reputation or the maturity. The assumption on the mixed derivative means that average reputation and reputation maturity are complementary: the demand effect of a higher score is stronger when the score is backed by a deeper track record.\\

The boundary condition $v(F,0)=0$ abstracts from exogenous discovery or exploration demand. Entry is therefore represented below through a strictly positive reset maturity $M_{\mathrm{reset}}>0$. Allowing a baseline flow independent of reputation would relax this assumption without changing the central incentive mechanism.

Mathematically, the instantaneous reward is
\begin{equation}
\pi(F_t,M_t, h_t, q_t) := \pi^{dev} (F_t,M_t) \, \chi_{\{h_t = 0\}} + \pi^{cont} (F_t,M_t,q_t) \, \chi_{\{h_t = 1\}}
\end{equation}
where $\chi$ is the indicator function, $\pi^{dev} (F_t,M_t)$ is the reward generated by deviation, and $\pi^{cont} (F_t,M_t,q_t)$ is the reward from continuing to behave honestly. Specifically, $\pi^{dev} (F_t,M_t)$ is described as:
\begin{equation} \label{eq:booty}
\pi^{dev} (F_t,M_t) := (\theta - \kappa^S) \, v(F_t,M_t) - \kappa^R
\end{equation}
$\theta\in(0,1]$ converts order flow into extractable value and can be interpreted as the average amount of user funds or economic surplus that the agent can appropriate per unit of activity. A larger $\theta$ therefore corresponds to environments in which each unit of trust-induced activity carries more stealable value. On the other side, there are two kinds of loss suffered by the malicious agent. There is a fixed component, $\kappa^R \ge 0$, related to re-entering the system with a new identity, capturing setup costs and the broader economic burden of rebuilding activity under a new or compromised identity, for example, rebuilding the reputation required to attract economically meaningful demand. $\kappa^R = 0$ represents the case where the agent chooses not to restore the identity, continuing to use it with a degraded reputation. Additionally, there is a component proportional to the handled volume that the market, platform, protocol, or contracting environment could require to stake and, eventually, slash. This is represented by $\kappa^S \in [0,\theta)$.  $\kappa^S$ is the effect of the slash in the S\&S policy. Therefore, $\kappa^S=0$ if and only if the protocol decides not to adopt this policy. Moreover, we restrict our analysis to $\kappa^S < \theta$, as it is the only economically relevant case.\\

Contrastingly, when the agent chooses to behave honestly, it has to set the quality level $q_t \in [0,1]$ that represents the effort (and the financial cost, computational resources, engineering practices, or technological standards) used to deliver the service. The upper bound $1$ denotes the maximal feasible service quality or technological standard that the agent can implement in one period. Quality is costly to provide but improves user outcomes and, therefore, future reputation. The cost is described by a function $c:q \in [0,1] \rightarrow [0,1]$. Furthermore, when the protocol chooses to implement the S\&S policy, the agent has to stake a quantity $\kappa^S v(F_t, M_t)$ computed as a fraction of the volume. This leads to an opportunity cost described by the function $c^o : \kappa^S v(F_t, M_t) \in [0,1] \rightarrow [0,1]$. Regarding the incomes, we assume the agent earns a fixed fee $f\in(0,\theta)$ per unit of order flow.

\begin{assumption} \label{ass:cost}
The function $c$ is increasing and strictly convex, while $c^o$ is increasing and convex. Both functions are equal to $0$ when evaluated at $0$.
Moreover, $(c^o)'(v) < f, \forall v$\footnote{This, in turn, implies $c^o(v) < v$. This technical assumption guarantees some desirable properties of the value function, such as the monotonicity in $F$ and $M$, and the fact that it cannot become negative.}.
\end{assumption}

Thus, under honest operation, the per-period profit is given by
\begin{equation}
\pi^{cont} (F_t,M_t,q_t) := f\,v(F_t,M_t) - c^o\left( \kappa^S \, v(F_t, M_t) \right) - c(q_t)
\label{eq:flowprofit}
\end{equation}

This specification encodes the economic mechanism of interest: reputation affects the scale of business through the state $(F_t,M_t)$ and therefore the current volume $v(F_t,M_t)$, while the agent can invest in quality to influence future reputation at a direct cost. Importantly, in this baseline specification, quality affects demand only through future reputation and not directly through contemporaneous demand. In other words, the chosen effort level affects future economic opportunities only through reputation. Furthermore, the opportunity cost makes the honest behavior less attractive. The protocol has, therefore, interests to reduce $c^o\left( \kappa^S v(F_t, M_t) \right)$, for a given $\kappa^S$. This may be accomplished by requiring the staking in illiquid tokens. Additionally, the protocol could remunerate the staked amount -- e.g., by re-investing it. However, a detailed analysis of the approaches to reduce the opportunity cost faced by the agent is beyond the aim of this paper.

Before introducing the dynamic evolution of reputation, it is useful to study the agent’s myopic incentive to deviate by comparing the current-period payoff from opportunistic deviation, $\pi^{dev}$, with that from honest continuation, $\pi^{cont}$. This comparison does not determine the optimal dynamic policy, since it ignores the future value of preserving or rebuilding reputational capital. Nevertheless, it provides a useful instantaneous benchmark for understanding how the model parameters affect the immediate attractiveness of opportunistic behavior. Moreover, assuming the future value coming from the honest behavior is greater than the future value after the deviation\footnote{The following section largely discusses this point.}, this comparison provides us with a necessary condition for the deviation to be profitable.

For a given state $(F,M)$ and a given quality choice $q$, the difference between the current-period deviation payoff and the current-period continuation payoff is
\begin{align}
\pi^{dev}(F,M)-\pi^{cont}(F,M,q) &= (\theta-\kappa^S) \, v(F,M)-\kappa^R - \left[ f\,v(F,M)-c^o\left(\kappa^S \, v(F,M)\right)-c(q) \right] \nonumber\\
&= (\theta-f-\kappa^S) \, v(F,M) -\kappa^R +c^o\left(\kappa^S \, v(F,M)\right) +c(q)
\label{eq:0418_1100}
\end{align}

Since $c(q)\geq 0$, we obtain the lower bound
\begin{equation}
\pi^{dev}(F,M)-\pi^{cont}(F,M,q) \geq (\theta-f-\kappa^S) \, v(F,M) -\kappa^R +c^o\left(\kappa^S \, v(F,M)\right).
\end{equation}
Hence, a sufficient condition for the instantaneous payoff from deviation to dominate the payoff from honest operation at state $(F,M)$ is
\begin{equation}
(\theta-f-\kappa^S) \, v(F,M) -\kappa^R +c^o\left(\kappa^S \, v(F,M)\right) \geq 0.
\end{equation}

To obtain a transparent benchmark, suppose that the opportunity cost is linear,
\begin{equation}
c^o(x)=a\,x, \qquad a\leq f
\end{equation}
The condition becomes
\begin{equation} \label{eq:instantaneous_deviation}
\left[ \theta-f-\kappa^S(1-a) \right] \, v(F,M) -\kappa^R > 0.
\end{equation}

The coefficient multiplying $v(F,M)$ determines how the instantaneous attractiveness of deviation varies with the amount of economic activity entrusted to the agent. As $\kappa^R \ge 0$, it immediately turns out that Eq. \eqref{eq:instantaneous_deviation} can be satisfied only if the $v(F,M)$ coefficient is positive: $\theta-f-\kappa^S(1-a)>0$. Since $v$ is increasing in both reputation and maturity, its maximum is attained at $(F,M)=(1,1)$. Hence, there exists at least one state for which the sufficient condition is satisfied whenever $\left[ \theta-f-\kappa^S(1-a) \right]v(1,1) -\kappa^R \geq 0$. Equivalently,
\begin{equation}
\kappa^S \leq \bar{\kappa}^S := \left[ \theta-f-\frac{\kappa^R}{v(1,1)} \right] \, (1-a)^{-1}.
\label{eq:0512_2047}
\end{equation}

The threshold $\bar{\kappa}^S$ should be interpreted as a \emph{myopic staking threshold}. When $\kappa^S<\bar{\kappa}^S$, sufficiently large values of $v(F,M)$ may make the instantaneous deviation payoff exceed the current-period payoff from honest operation. Thus, honest behavior cannot be inferred solely on the basis of the current-period payoff comparison. Conversely, when $\kappa^S>\bar{\kappa}^S$, the sufficient condition above cannot be satisfied at any state.

Several comparative-static implications are immediate. A larger extractable-value parameter $\theta$ raises $\bar{\kappa}^S$, because a larger fraction of the entrusted economic flow can be appropriated through deviation. A larger fee rate $f$ lowers the threshold by increasing the current return to honest operation. A larger fixed post-deviation cost $\kappa^R$ also lowers $\bar{\kappa}^S$, since misconduct becomes less attractive even before accounting for future reputational losses. Finally, a larger opportunity-cost parameter $a$ increases the threshold because staking becomes more costly under honest operation.

Importantly, Eq.~\eqref{eq:0512_2047} is not a sufficient characterization of the agent's optimal dynamic policy. The instantaneous comparison ignores the future income associated with preserving the current reputation, as well as the economic consequences of moving to a penalized reputational state after deviation. These intertemporal effects are introduced through the state-transition dynamics below and are subsequently incorporated into the agent's dynamic optimization problem.

\section{Dynamic Optimization and Reputational Discipline}
\label{sec:dynamic}

This section formulates the agent's dynamic optimization problem
using the Bellman framework and establishes its well-posedness. In
particular, we prove the existence and uniqueness of the value function and
show that it is monotonic in the reputational state, establishing reputation
as an economically valuable asset.

\subsection{Value Function and Dynamic Objective}

The agent's objective is to maximize the total economic value generated over its lifetime. Since interactions with counterparties are repeated, the relevant decision is inherently dynamic: the agent must compare the immediate benefit of its current action with its impact on future reputation and, therefore, on future earning opportunities. At every point in time, the agent faces two decisions. First, it decides whether to continue operating honestly or to behave opportunistically. Second, conditional on remaining honest, it chooses the level of service quality to provide. Both decisions affect the future evolution of the reputation state and therefore the stream of future profits.

Starting from state $(F_0,M_0)=(F,M)$, the agent chooses a sequence of controls
\[
\{(h_t,q_t)\}_{t\in\mathbb{N}},
\]
where $h_t\in\{0,1\}$ determines whether the agent continues honestly or deviates, and $q_t\in[0,1]$ is relevant only when $h_t=1$. The control at time $t$ is supposed to be chosen after observing the state at time $t$. Nonetheless, we neglect this technical aspect due to the deterministic nature of our setting. Future profits are discounted by a factor $\beta\in(0,1)$. The coefficient $\beta$ depends and can be calibrated on the client arrival frequency -- $\beta$ is closer to $1$ as the arrival frequency increases.

The value function, therefore, is defined as the maximum discounted lifetime payoff attainable from a given reputation state,
\begin{equation}
\label{eq:valuefunction}
V(F,M) := \sup_{\{(h_t,q_t)\}_{t\in\mathbb{N}}} \sum_{t=0}^{\infty} \beta^t \, \pi(F_t,M_t,h_t,q_t) \qquad\qquad with \ (F_{t+1}, M_{t+1}) = \begin{cases} T(F_t, M_t, q_t) & if \ h_t = 1 \\ R(F_t, M_t) & if \ h_t = 0 \end{cases}
\end{equation}

The dynamic programming formulation naturally decomposes the current decision into two alternative branches: deviating immediately or continuing honest operation. The value associated with each branch consists of two components: the current-period payoff generated by the chosen action and the discounted value of the reputation state reached in the following period.

To characterize the decision at the current state, it is useful to separate the deviation and continuation branches. For any candidate future-value function $W$, define
\begin{align}
V^{dev}_W(F,M) & := (\theta-\kappa^S) \, v(F,M) -\kappa^R + \beta \, W\big(R(F,M)\big) \label{eq:deviation_branch} \\
V^{cont}_W(F,M) & := f\,v(F,M) - c^o\left(\kappa^S \, v(F,M)\right) + \max_{q\in[0,1]}U_W(F,M;q) \label{eq:continuation_branch} \\
U_W(F,M;q) & := -c(q) + \beta \, W\big(T(F,M,q)\big). \label{eq:continuation_kernel}
\end{align}
The value function is therefore written in terms of the Bellman equation:
\begin{equation} \label{eq:valuefunction_bellman}
V(F,M) = \max \left\{ V^{dev}_V(F,M), V^{cont}_V(F,M) \right\}.
\end{equation}

The deviation branch combines the immediate payoff from opportunistic behavior with the continuation value of future activity after misconduct. The first two terms capture the current net gain from deviation, including the extractable value of the interaction, the loss of the required stake, and any additional fixed penalty. The continuation term $
\beta \, W\big(R(F,M)\big)$ represents the discounted lifetime value of operating from the post-deviation state. Thus, if misconduct resets the agent to an initial reputation state, this branch includes both the one-shot gain from deviation and the value of rebuilding reputation thereafter.

The continuation branch instead describes the value of preserving honest operation. It consists of the current fee income, the opportunity cost of locking the required stake, the cost of providing service quality, and the discounted continuation value associated with the reputation state reached through the honest transition \(T(F,M,q)\). The maximization over \(q\) captures the trade-off between the current cost of supplying higher quality and the future value generated by a stronger reputation.

Notice that the previously discussed fixed-reset formulation is nested in Eq.~\eqref{eq:deviation_branch}. If
\[
R(F,M) \equiv (F_{\mathrm{reset}},M_{\mathrm{reset}}),
\]
then
\[
W(R(F,M)) = W(F_{\mathrm{reset}},M_{\mathrm{reset}})
\]
is independent of the current state. By contrast, under continued operation after misconduct, the future value following deviation remains state-dependent.

\subsection{Solvability and Shape Properties of the Bellman Problem}
\label{sec:solv}

Let $C_b(\mathcal{X})$ denote the Banach space of bounded continuous functions from $\mathcal{X}$ to $\mathbb{R}$, endowed with the supremum norm
\[
\|W\|_\infty = \sup_{(F,M)\in\mathcal{X}}|W(F,M)|.
\]
Since $\mathcal{X}$ is compact, every $W\in C_b(\mathcal{X})$ is also uniformly continuous. For any $W\in C_b(\mathcal{X})$, define the Bellman operator $\mathcal{B}:C_b(\mathcal{X})\rightarrow C_b(\mathcal{X})$ by
\begin{equation}
\label{eq:bellman_operator}
(\mathcal{B}W)(F,M) := \max \left\{ V^{dev}_W(F,M), V^{cont}_W(F,M) \right\}.
\end{equation}

The following proposition establishes that the dynamic optimization problem is well posed. Specifically, it proves the existence and uniqueness of the value function together with its main shape properties, namely monotonicity with respect to reputational capital and the staking requirement, and non-negativity. These results provide the theoretical foundation for the subsequent analysis.

\begin{proposition}[Existence, Monotonicity, and Sign]
\label{prop:0509_0709}
Assume \ref{ass:update}--\ref{ass:cost}. Then there exists a unique, uniformly continuous value function $V$ satisfying Eq.~\eqref{eq:valuefunction_bellman}. Furthermore:
\begin{itemize}[label=$\triangleright$]
    \item $V$ is non-decreasing in $F$ and $M$;
    \item $V$ is non-increasing in $\kappa^S$;
    \item $V(F,M)\geq0$ for every $(F,M)\in\mathcal{X}$.
\end{itemize}
\end{proposition}
The proof of this result is in Appendix \ref{app_proofs_dynamic}.

\begin{remark}
\label{rem:general_R}
Proposition \ref{prop:0509_0709} does not require a particular post-deviation mechanism. The existence and uniqueness result requires only continuity of the post-deviation transition, while the monotonicity of the value function in reputational capital additionally follows from the order-preserving property of $R$.
\end{remark}

\begin{remark}
    The monotonicity of $V$ as a function of $F$ and $M$ gives us a very important economic message we were looking for: the reputation is economic capital. More reputation means agent's larger value. This is the baseline for interpreting the results of our model.
\end{remark}

\subsection{Temptation Function and Post-deviation Mechanisms}
\label{sec:temptation}

The central economic question of the model is whether the agent prefers to
preserve its reputational capital or to exploit the economic value currently
associated with it through opportunistic deviation. We therefore define the
continuation region $\mathcal{C}$, the deviation region $\mathcal{D}$, and the
indifference region $\mathcal{I}$ as
\begin{align}\label{eq:regions}
\mathcal{C}
&=
\left\{
(F,M)\in\mathcal{X}
\;:\;
V^{dev}(F,M)\le V^{cont}(F,M)
\right\},
\\
\mathcal{D}
&=
\left\{
(F,M)\in\mathcal{X}
\;:\;
V^{dev}(F,M)\geq V^{cont}(F,M)
\right\},
\nonumber
\\
\mathcal{I}
&=
\left\{
(F,M)\in\mathcal{X}
\;:\;
V^{dev}(F,M)=V^{cont}(F,M)
\right\}
=
\mathcal{C}\cap\mathcal{D}.
\nonumber
\end{align}

We summarize the relative attractiveness of the two actions introducing the \emph{temptation function}
\begin{equation}
\label{eq:temptation}
\Delta(F,M) := V^{dev}(F,M)-V^{cont}(F,M).
\end{equation}
Thus,
\[
\mathcal{C} = \Delta^{-1}\big((-\infty,0]\big), \qquad \mathcal{D} = \Delta^{-1}\big([0,\infty)\big),
\]
while
\[
\mathcal{I} = \Delta^{-1}(\{0\}).
\]

Using Eqs.~\eqref{eq:deviation_branch}--\eqref{eq:continuation_kernel}, the
temptation function can be written as
\begin{align}
\Delta(F,M) =&\;
(\theta-\kappa^S-f) \, v(F,M)
-\kappa^R
+
c^o\big(\kappa^S \, v(F,M)\big)
\nonumber\\
&+
\beta \, V\big(R(F,M)\big)
-
\max_{q\in[0,1]}
\left\{
-c(q)
+
\beta \, V\big(T(F,M,q)\big)
\right\}.
\label{eq:temptation_explicit}
\end{align}

Equation \eqref{eq:temptation_explicit} provides a transparent decomposition of the temptation to deviate. The first line captures the immediate consequences of misconduct: $(\theta-\kappa^S-f) \, v(F,M)$ is the net value that can be extracted through deviation, $-\kappa^R$ is the direct reputational penalty, and $c^o(\kappa^S \, v(F,M))$ is the operational cost associated with the stake. The second line captures the dynamic trade-off: $\beta \, V(R(F,M))$ is the continuation value after deviation, while the maximization term is the value of remaining honest, accounting for the optimal quality choice. Reputation, therefore, affects temptation through two channels: it increases both the current value available for opportunistic extraction and the future value that is forfeited by deviating. The current reputation allows us to understand when the reputational capital is preserved and when it is abandoned for deviation.

\paragraph{Comparison of post-deviation mechanisms.}

We now compare the two post-deviation mechanisms introduced above. The
relevant difference between them is the amount of reputational capital that
survives misconduct. Under $R^{same}$, the agent remains under the same
identity and therefore retains part of its pre-deviation reputation and
maturity. Under $R^{new}$, instead, the agent is forced to restart from the
fixed state $(F_{\mathrm{reset}},M_{\mathrm{reset}})$.

To isolate the direct incentive effect of this difference, we compare the two
mechanisms while holding the future-value function fixed. For any candidate
future-value function $W$ and post-deviation transition $R$, define
\begin{equation}
\Delta_W^R(F,M)
:=
V^{dev}_{W,R}(F,M)-V^{cont}_W(F,M).
\end{equation}
The continuation branch is identical under the two mechanisms, as is the
current payoff generated by deviation. Hence, the only difference between
$\Delta_W^{same}$ and $\Delta_W^{new}$ comes from the value of the state reached
after misconduct.

Under the same-identity mechanism,
\[
R^{same}(F,M)
=
\left(
(1-\lambda)F+\lambda s^{dev},
\,
(1-\rho)M+\rho\sigma(F,M)
\right),
\]
whereas under identity replacement,
\[
R^{new}(F,M)
=
(F_{\mathrm{reset}},M_{\mathrm{reset}}).
\]
Therefore, whenever
\[
(1-\lambda)F+\lambda s^{dev}\geq F_{\mathrm{reset}}
\qquad\text{and}\qquad
(1-\rho)M+\rho\sigma(F,M)\geq M_{\mathrm{reset}},
\]
the agent retains weakly more reputational capital after deviation under
$R^{same}$ than under $R^{new}$, that is,
\[
R^{same}(F,M)\succeq R^{new}(F,M).
\]
Since a larger reputational state generates a weakly larger future value, this
immediately implies that preserving the identity makes deviation more
attractive.

\begin{proposition}[Post-deviation comparison]
\label{prop:R_temptation_comparison}
Let $W$ be non-decreasing in both $F$ and $M$. Consider a state $(F,M)$ such
that the post-deviation state under same-identity continuation weakly dominates
the reset state,
\begin{equation}
\label{eq:0821_2318}
R^{same}(F,M)\succeq R^{new}(F,M),
\end{equation}
or, equivalently,
\begin{equation}
F
\geq
\frac{F_{\mathrm{reset}}-\lambda s^{dev}}{1-\lambda},
\qquad
M
\geq
\frac{M_{\mathrm{reset}}-\rho\sigma(F,M)}{1-\rho}.
\end{equation}
Then
\begin{equation}
\Delta_W^{same}(F,M)-\Delta_W^{new}(F,M)
=
\beta
\left[
W\big(R^{same}(F,M)\big)
-
W\big(R^{new}(F,M)\big)
\right]
\geq0.
\end{equation}
Hence, holding the continuation-value environment fixed, a post-deviation
mechanism that preserves a larger fraction of reputational capital makes
opportunistic deviation weakly more attractive.
\end{proposition}
The proof of this result is in Appendix \ref{app_proofs_dynamic}.

\begin{remark}
Proposition \ref{prop:R_temptation_comparison} isolates the direct effect of the post-deviation mechanism. When the reputational capital is sufficiently developed, holding the continuation-value environment fixed, preserving a larger fraction of reputational capital after misconduct makes opportunistic deviation more attractive. Identity replacement, therefore, provides stronger reputational discipline precisely because it severs a larger fraction of the link between the pre-deviation state and future economic value.

The proposition deliberately compares the two mechanisms for a common candidate value function $W$. In the fully endogenous equilibrium comparison, the value function itself generally differs across the two environments -- that is, $V^{new}$ obtained from $R^{new}$ is different than $V^{same}$ coming from $R^{same}$. Therefore, a pointwise ranking of $\Delta^{same}$ and $\Delta^{new}$ at their respective fixed points does not follow from Proposition \ref{prop:R_temptation_comparison} alone, as it assumes $W$ to be the same both in the $R^{new}$ and $R^{same}$ cases. The result instead identifies the direct mechanism through which post-deviation reputational persistence affects the incentive to deviate.
\end{remark}

Proposition \ref{prop:R_temptation_comparison} shows that the effectiveness of a reputation system depends not only on how reputation is accumulated, but also on how it is treated following misconduct. Policies that force an agent to restart under a new identity destroy a larger fraction of its accumulated reputational capital, thereby reducing the continuation value available after deviation and strengthening the incentive to remain honest. By contrast, if a substantial part of the reputation survives misconduct, opportunistic behavior becomes more attractive because the agent retains a larger fraction of its future earning opportunities.


\subsection{Monotonicity of the Temptation}
We previously defined the temptation function $\Delta$. Now, we study its monotonicity in the reputational capital -- that is, under which conditions it holds $(F_1, M_1) \preceq (F_2, M_2) \implies$ $\Delta(F_1, M_1) \le \Delta(F_2, M_2)$ or $\Delta(F_1, M_1) \ge \Delta(F_2, M_2)$. To achieve this goal, we first need to derive conditions for the Lipschitz continuity of the value function. To this end, we impose additional regularity assumptions on the state-transition maps and derive Lipschitz bounds for the value function.

\begin{assumption}[Lipschitz post-deviation transition]
\label{ass:R_lipschitz}
We restrict the general post-deviation map of Assumption \ref{ass:deviation_transition} to the class $R(F,M) = \big( R_F(F), R_M(F,M) \big)$ such that $R_F$ and $R_M$ are Lipschitz continuous. In particular, there exist non-negative constants $L_{R,F}$, $L_{R,MF}$, and $L_{R,MM}$ such that
\begin{align*}
|R_F(F_1)-R_F(F_2)| &\leq L_{R,F}|F_1-F_2|, \\
|R_M(F_1,M)-R_M(F_2,M)| &\leq L_{R,MF}|F_1-F_2|, \\
|R_M(F,M_1)-R_M(F,M_2)| &\leq L_{R,MM}|M_1-M_2|.
\end{align*}
\end{assumption}

We note that assumption \ref{ass:R_lipschitz} is satisfied by $R^{new}$ and, when $\sigma$ is Lipschitz continuous, also by $R^{same}$. The following result provides the Lipschitz characterization.

\begin{proposition}[Lipschitz Continuity]
\label{prop:lip}
Assume \ref{ass:update}--\ref{ass:cost} and Assumption \ref{ass:R_lipschitz}. Assume furthermore that $v$ and $\sigma$ are Lipschitz continuous with constants $L_v$ and $L_\sigma$, respectively. Define $A_M := 1-\rho+\rho L_\sigma$. Suppose that
\begin{equation}
\label{eq:0508_1500}
\beta A_M<1, \qquad\qquad \beta L_{R,MM}<1, \qquad\qquad \beta L_{R,F}<1.
\end{equation}

Then the value function $V$ is Lipschitz continuous separately in $M$ and
$F$. In particular, valid Lipschitz constants are
\begin{align}
\label{eq:0508_1501}
L_M &= L_v \max \left\{ \frac{\theta-\kappa^S}{1-\beta \, L_{R,MM}}, \, \frac{f}{1-\beta \, A_M} \right\} \\
L_F &= \max \left\{ \frac{(\theta-\kappa^S) \, L_v + \beta \, L_{R,MF} \, L_M}{1-\beta \, L_{R,F}}, \, \frac{f \, L_v + \beta \, \rho \, L_\sigma \, L_M}{1-\beta \, (1-\lambda)} \right\} \nonumber
\end{align}
\end{proposition}
The proof of this result is in Appendix \ref{app_proofs_dynamic}.

\begin{remark}
Under the identity replacement mechanism $R^{new}$, the $R$- related Lipschitz constants $L_{R,F}$, $L_{R,MF}$, and $L_{R,MM}$ are $0$. Under different mechanisms, the additional terms coming by these constants have a direct economic interpretation. They quantify the extent to which reputational capital survives opportunistic behavior and therefore continues to affect the agent's future value after a deviation. If we make the additional assumption of linear opportunity cost,
\[
c^o(\kappa^S v)=a\kappa^S v,
\qquad
a<f,
\]
the continuation-branch bounds can be sharpened because
\[
f\,v-c^o(\kappa^S v)
=
(f-a\kappa^S)v.
\]
Equation \eqref{eq:0508_1501} therefore becomes
\begin{align}
\label{eq:0508_1504}
L_M
&=
L_v
\max
\left\{
\frac{\theta-\kappa^S}
{1-\beta L_{R,MM}},
\,
\frac{f-a\kappa^S}
{1-\beta(1-\rho+\rho L_\sigma)}
\right\},
\\
L_F
&=
\max
\left\{
\frac{
(\theta-\kappa^S)L_v
+
\beta L_{R,MF}L_M
}
{1-\beta L_{R,F}},
\,
\frac{
(f-a\kappa^S)L_v
+
\beta\rho L_\sigma L_M
}
{1-\beta(1-\lambda)}
\right\}.
\nonumber
\end{align}
\end{remark}

We now exploit the Lipschitz continuity to study how temptation varies with the amount of accumulated reputational capital. The economic effect is in principle ambiguous. A stronger reputation or a more mature track record increases the current economic activity that can potentially be extracted, but it also increases the future value associated with remaining in the market. Additionally, when part of the reputational state survives misconduct, a better current state may also increase the continuation value following deviation. The following result provides sufficient conditions under which the overall
temptation to deviate increases with reputational capital.

\begin{proposition}[General Monotonicity]
\label{prop:temptation_monotonicity}
Assume \ref{ass:update}--\ref{ass:cost} and the additional assumptions of
Proposition \ref{prop:lip}. Then, for every $\bar{F},\bar{M}>0$, if
\begin{equation}
\label{eq:0804_2117}
\begin{cases} \big[ \theta-f-\kappa^S \big] \, \partial_F v(1,\bar{M}) \geq \beta \, L_F \, (1-\lambda) + \beta \, \rho \, L_M \, L_\sigma & \mathbf{(F)} \\[0.3cm]
\big[ \theta-f-\kappa^S \big] \, \partial_M v(\bar{F},1) \geq \beta \, L_M \, (1-\rho) + \beta \, \rho \, L_M \, L_\sigma & \mathbf{(M)} \end{cases}
\end{equation}
then the temptation function $\Delta(F,M)$ is non-decreasing in the corresponding variable on $[\bar{F},1]\times[\bar{M},1]$.
\end{proposition}

The proof is in Appendix \ref{app_proofs_dynamic}.

\begin{remark}
Proposition \ref{prop:temptation_monotonicity} concerns the direction in which temptation changes with reputation, rather than the sign of temptation itself. Thus, the proposition does not by itself imply that deviation is optimal at high values of $F$ or $M$; it implies that, under Eq.~\eqref{eq:0804_2117}, increasing the corresponding reputational component cannot reduce the relative attractiveness of deviation. The conditions in Eq.~\eqref{eq:0804_2117} are sufficient rather than necessary. They compare the marginal increase in the economic value that can be extracted through deviation,
\[
\big[\theta-f-\kappa^S\big]\partial_x v, \qquad x\in\{F,M\},
\]
with an upper bound on the additional continuation value generated by a stronger reputation under honest behavior. When the former effect dominates, the temptation to deviate is non-decreasing in the corresponding state variable. The conditions are conservative because the positive post-deviation
continuation term
\[
\beta \, \left[ V\big(R(F,M_1)\big) - V\big(R(F,M_2)\big) \right]
\]
is omitted when constructing the lower bound for the change in temptation. Hence, when the post-deviation mechanism preserves part of the current reputational capital, temptation may be increasing even when Eq.~\eqref{eq:0804_2117} is not satisfied.

Monotonicity also provides a useful characterization of the deviation region. For example, if $\Delta(F,M)$ is non-decreasing in $M$ for fixed $F$ and there exists a threshold $M^\star(F)$ such that
\[
\Delta(F,M^\star(F))=0,
\]
then
\[
M<M^\star(F) \quad\Longrightarrow\quad \Delta(F,M)\leq0,
\]
whereas
\[
M>M^\star(F) \quad\Longrightarrow\quad \Delta(F,M)\geq0.
\]
Thus, whenever such a crossing exists, monotonicity implies a threshold structure for the continuation and deviation regions. An analogous argument applies to $F$.

The restriction to $\bar{F},\bar{M}>0$ avoids boundary regions in which the relevant marginal effect of the volume function may become degenerate or ill-behaved. For example, under the Cobb--Douglas specification
\[
v(F,M) = F^\eta \, M^\zeta, \qquad \eta,\zeta\in(0,1),
\]
the derivatives of $v$ may become singular near the boundary of the state space.
\end{remark}

Proposition \ref{prop:temptation_monotonicity} highlights a not very intuitive implication of the model. One might naturally expect that accumulating more reputational capital always strengthens the incentive to behave honestly, since a more reputable agent has more future business to lose. This is indeed one force at work: a better reputation increases the value of remaining in the market and therefore raises the cost of misconduct.

However, the opposite force is also present. A better reputation increases the economic activity associated with the agent through $v(F,M)$, so that a more reputable agent also has more value available to extract opportunistically at the moment of deviation. The proposition identifies sufficient conditions under which the current gain from exploiting a stronger reputation grows faster than the future benefit of preserving it through honest behavior. In that case, the temptation to deviate increases with accumulated reputation. Thus, reputational capital plays a dual role: it is both an asset that the agent may wish to preserve through honest behavior and an asset whose accumulated economic value can be exploited through misconduct. Ultimately, having identified these two opposing forces, in the next section, we aim at making an explicit characterization of when one dominates over the other.

\section{Small-update Regime}
\label{sec:sur}

The previous analysis applies to a general post-deviation transition and therefore encompasses both continued operation under the same identity and identity replacement. In this section, we specialize the analysis to the identity-reset mechanism, in which detected misconduct forces the agent to restart from a fixed reputational state,
\[
R(F,M)=R^{new}(F,M)=\big(F_{\mathrm{reset}},M_{\mathrm{reset}}\big).
\]
This specialization removes the state dependence of the post-deviation continuation value and makes it possible to obtain a sharper analytical characterization of the continuation and deviation incentives.

Our main objective is to make the general monotonicity result of Proposition~\ref{prop:temptation_monotonicity} more explicit in terms of the primitive parameters of the model. The sufficient conditions derived in the previous section establish when temptation increases with reputational capital, but they depend on endogenous Lipschitz bounds for the value function and therefore do not immediately reveal how parameters such as the extractable-value share, the discount factor, fees, staking requirements, and staking opportunity costs determine the direction of the effect.

To obtain such a characterization, we focus on the small-update regime, in which the reputation-update parameters $\lambda$ and $\rho$ are bounded by arbitrarily small positive thresholds $\bar{\lambda}$ and $\bar{\rho}$. Economically, this corresponds to an environment in which reputational capital is highly persistent and each additional interaction changes the reputational state only gradually. In the limiting case $\lambda=\rho=0$, the Bellman problem simplifies sufficiently to yield a closed-form expression for the temptation function. This allows us to identify directly, in terms of the economic primitives, whether additional reputational capital strengthens discipline or instead makes opportunistic deviation increasingly attractive. By uniform continuity, these qualitative results extend to sufficiently small positive values of $\lambda$ and $\rho$.

The mathematical framework is as follows. We extend the value function domain to $\mathcal{X} \times [0,1]^2$, by considering $4$-dimensional inputs: $(F, M, \lambda, \rho)$. Thus, $(\lambda, \rho)$ are viewed as state variables rather than environmental constants. The existence and uniform continuity of the value function $V:\mathcal{X} \times [0,1]^2 \rightarrow \mathbb{R}$ is proven as previously done. Then, we focus on the small-update regime $\mathcal{X} \times [0,\bar{\lambda}] \times [0,\bar{\rho}]$. The idea is to infer information on the points satisfying a given property, in the small-update regime, by simply looking at what happens when $\lambda=0$ and $\rho=0$. The main advantage of this approach is that we get rid of the iterative component of the value function. Indeed, $\forall (F,M) \in \mathcal{X}$, the state-update function coincides with the identity: $T(F,M, q) = (F,M)$, $\forall q \in [0,1]$ and the continuation kernel is trivially maximized by $q=0$: $\max_q U(F, M, q) = \beta \, V(F, M)$. \\
\linebreak
Specifically, assume to be interested in properties represented as $\phi(F,M, \lambda, \rho) \ge 0$, for some uniformly continuous function $\phi$. The target is to identify the subset of points $\mathcal{A}_{\lambda, \rho, 0} \subset \mathcal{X}$ defined as $\mathcal{A}_{\lambda, \rho, 0} = \big\{ (F,M) \in \mathcal{X} \text{ s.t. } \phi(F,M, \lambda, \rho) \ge 0 \big\}$. In the small-update regime, we can study the behavior of $\phi(F,M, 0, 0) \ge \xi$, for some $\xi > 0$, and then exploit the uniform continuity for claiming the existence of $\bar{\lambda}$ and $\bar{\rho}$ such that, $\forall (\lambda, \rho) \in [0,\bar{\lambda}] \times [0,\bar{\rho}]$, we have $\mathcal{A}_{0, 0, \xi} \subset \mathcal{A}_{\lambda, \rho, 0}$. In this way, we obtain explicit expressions for the continuation and deviation value functions. These, in turn, yield a closed-form characterization of the temptation function and allow the monotonicity conditions derived in the previous section to be expressed directly in terms of the primitive parameters of the model. By uniform continuity, these qualitative conditions extend to the small-update regime for sufficiently small positive values of $\lambda$ and $\rho$. This characterization also makes it possible to isolate the contribution of each primitive parameter, providing a transparent basis for economic interpretation and policy analysis.

\begin{proposition}[Explicit form of the value functions] \label{prop:0825_1616}
    Assume \ref{ass:update} to \ref{ass:volume}, and $\lambda = \rho = 0$. Then, we define $\tilde{V}^{dev}$ and $\tilde{V}^{cont}$ as the value of the branches assuming that, at time $0$, the agent chooses to deviate or to honestly behave. $\tilde{V}^{dev}$ and $\tilde{V}^{cont}$ read:
    \begin{align} \label{eq:0417_1410}
    \tilde{V}^{dev}(F,M,0,0) &= (\theta - \kappa^S) \, v(F, M) + \frac{1}{1-\beta} \, \max \begin{cases} (\theta - \kappa^S) \, \beta v_{reset} - \kappa^R \\ \beta \big[ f \, v_{reset} - c^o(\kappa^S \, v_{reset}) \big] - (1 - \beta) \, \kappa^R \end{cases} \\
    \tilde{V}^{cont}(F,M,0,0) &= \frac{1}{1-\beta} \, \bigg[ fv(F,M) - c^o(\kappa^S \, v(F,M)) \bigg] \nonumber
    \end{align}
\end{proposition}

The proof of this result is in Appendix \ref{app_proofs_sur}.

\begin{remark}
    Note that $\tilde{V}^{dev}(F,M,0,0)$ is always equal to $V^{dev}(F,M,0,0)$. Instead, $\tilde{V}^{cont}(F,M,0,0)$ is equal to $V^{cont}(F,M,0,0)$ only if the agent honestly behaves. Furthermore, the dependence of $\tilde{V}^{dev}(F,M,0,0)$ and $\tilde{V}^{cont}(F,M,0,0)$ on $F$ and $M$ is through the volume only. Furthermore, the previous result does not hold in the small-update regime, but only in the hyperspace defined by $\lambda=\rho=0$. In other words, we don't have to think about it as a property of the value function, but rather as a trick to simplify calculations.
\end{remark}

\begin{remark}
    We define the simplified temptation function $\tilde{\Delta}(F,M,0,0)$ as:
    \begin{align} \label{eq:2608_1402}
        \tilde{\Delta}(F,M,0,0) =& \frac{1}{1-\beta} \, \big\{ \big[(\theta-\kappa^S) \, (1-\beta)-f\big] \, v(F,M) +c^o(\kappa^S \, v(F,M)) \big\} \\
        &+ \frac{1}{1-\beta} \, \begin{cases}
            \beta\, (\theta-\kappa^S) \, v_{reset}-\kappa^R, & \text{if } (\theta-\kappa^S-f) \, v_{reset} \ge \kappa^R-c^o(\kappa^S \, v_{reset}), \\
            \beta \, \big[f \, v_{reset}-c^o(\kappa^S \, v_{reset})\big] -(1-\beta) \, \kappa^R, & \text{otherwise}. \end{cases} \nonumber
    \end{align}
To study the temptation of the agent, $\tilde{\Delta}$ is equivalent to $\Delta$: their signs coincide. Assuming a linear opportunity cost, $c^o(\kappa^S \, v(F,M))=a \, \kappa^S \, v(F,M)$ with $a<f$, the closed-form expression for $\tilde{\Delta}(F,M,0,0)$ implies an explicit threshold in the reputation-dependent volume $v(F,M)$. In the exact zero-update limit, this threshold separates the continuation and deviation regions. By uniform continuity, for sufficiently small positive values of $\lambda$ and $\rho$, states that lie at least a margin $\xi>0$ away from this threshold preserve their classification. Hence, rather than characterizing the whole regions $\mathcal{C}$ and
$\mathcal{D}$ exactly, we identify subsets that are guaranteed to remain in
the continuation and deviation regions throughout the small-update regime.
To this end, define
\[
A(\kappa^S)
=
(\theta-\kappa^S)(1-\beta)-(f-a\kappa^S),
\]
and
\[
T(\kappa^S)
=
\begin{cases}
\displaystyle
\frac{\kappa^R-\beta(\theta-\kappa^S)v_{\mathrm{reset}}}
{A(\kappa^S)},
&
\text{if }
\big[\theta-(1-a)\kappa^S-f\big]v_{\mathrm{reset}}
\ge \kappa^R,
\\[4mm]
\displaystyle
\frac{(1-\beta)\kappa^R-\beta(f-a\kappa^S)v_{\mathrm{reset}}}
{A(\kappa^S)},
&
\text{otherwise}.
\end{cases}
\]

If $A(\kappa^S)>0$, then
\begin{align}
\label{eq:0512_2045}
\mathcal{D}
&\supseteq
v^{-1}\big(T(\kappa^S)+\xi,\infty\big),
\\
\mathcal{C}
&\supseteq
v^{-1}\big(-\infty,T(\kappa^S)-\xi\big).
\nonumber
\end{align}
If instead $A(\kappa^S)<0$, the orientation of the threshold is reversed:
\begin{align}
\mathcal{D}
&\supseteq
v^{-1}\big(-\infty,T(\kappa^S)-\xi\big),
\\
\mathcal{C}
&\supseteq
v^{-1}\big(T(\kappa^S)+\xi,\infty\big).
\nonumber
\end{align}
When $A(\kappa^S)=0$, the zero-update expression is independent of
$v(F,M)$, and the classification depends only on the sign of the corresponding
constant term.
\end{remark}

We now study how the temptation function evolves as reputational capital accumulates. Our objective is to identify the parameter regimes under which reputation reinforces honest behavior or, conversely, makes opportunistic deviation increasingly attractive. The following result provides explicit conditions on the primitives of the model determining which of these two effects dominates in the small-update regime. The proof is in Appendix \ref{app_proofs_sur}.

\begin{proposition}[Small-update regime monotonicity] \label{prop:0825_1626}
    Assume \ref{ass:update} to \ref{ass:volume}. Assume that one of the following conditions hold:
    \begin{align} \label{eq:0417_1344}
        (\theta - \kappa^S) \, (1-\beta) - f \ge \xi & \qquad\qquad\text{\textbf{Increasing}} \\
        (\theta - \kappa^S) \, (1-\beta) - f + \kappa^S \, \le - \xi & \qquad\qquad\text{\textbf{Decreasing}} \nonumber
    \end{align}
    for some $\xi > 0$. With linear opportunity cost $c^o(v) = a \, v$, these conditions can be sharpened as:
    \begin{equation}
        (\theta - \kappa^S) \, (1-\beta) - f + a \, \kappa^S \gtrless \pm \xi
    \end{equation}
    Then, $\forall \bar{\delta}, \bar{F}, \bar{M} \in (0,1)$, there exists a small-update regime $[0, \bar{\lambda}) \times [0, \bar{\rho})$ such that $\forall F_2 \in [\bar{F},1]$, $\forall M_2 \in [\bar{M},1]$, $\forall F_1 \in [0, F_2 \, \bar{\delta}]$, and $\forall M_1 \in [0, M_2 \, \bar{\delta}]$, it holds:
    \begin{align}
        \tilde{\Delta}(F_1,M_2, \lambda, \rho) \le \tilde{\Delta}(F_2,M_2, \lambda, \rho) & \qquad\qquad \Delta(F_2,M_1, \lambda, \rho) \le \tilde{\Delta}(F_2,M_2, \lambda, \rho) & \qquad\qquad\text{\textbf{Increasing}} \\
        \tilde{\Delta}(F_1,M_2, \lambda, \rho) \ge \tilde{\Delta}(F_2,M_2, \lambda, \rho) & \qquad\qquad \tilde{\Delta}(F_2,M_1, \lambda, \rho) \ge \tilde{\Delta}(F_2,M_2, \lambda, \rho) & \qquad\qquad\text{\textbf{Decreasing}} \nonumber
    \end{align}
\end{proposition}

\begin{remark}
Proposition \ref{prop:0825_1626} characterizes how the temptation to deviate changes with reputational capital in the small-update regime. The result does not identify whether deviation is optimal at a given state; rather, it determines whether increasing $F$ or $M$ makes deviation relatively more or less attractive. In the limit $\lambda=\rho=0$, this direction is governed by the sign of the coefficient multiplying the reputation-dependent volume term in $\tilde{\Delta}(F,M,0,0)$. With linear opportunity cost, $c^o(v)=av$, the relevant quantity is
\[
(\theta-\kappa^S) \, (1-\beta)-f+a \, \kappa^S.
\]
When this term is sufficiently positive, the additional economic value created by a stronger reputation increases the payoff from deviation faster than the value of honest continuation, so that temptation increases with reputation. When it is sufficiently negative, the opposite force dominates and temptation decreases.

The expression also provides a direct economic interpretation of the primitives governing these two forces. A larger $\theta$ increases the fraction of economic activity that can be appropriated through misconduct and therefore pushes the system toward increasing temptation. A larger honest fee $f$ has the opposite effect, since it raises the return from preserving reputation and continuing to operate. A larger discount factor $\beta$ increases the relative importance of future honest earnings and therefore strengthens reputational discipline. The role of the stake $\kappa^S$ is more subtle. A larger stake reduces the net value that can be captured through deviation, but it also raises the opportunity cost of capital borne during honest operation. Under linear opportunity costs, the net effect of increasing $\kappa^S$ on the monotonicity condition is governed by
\[
\kappa^S \, \big[a-(1-\beta)\big].
\]
When $a<1-\beta$, increasing the stake makes the sufficient monotonicity condition more difficult to satisfy, so that the system may transition from a regime in which temptation increases with reputation to one in which reputation instead strengthens discipline. Conversely, when $a>1-\beta$, the opportunity cost of maintaining the stake dominates its disciplinary effect, so that increasing $\kappa^S$ shifts incentives in the opposite direction. Finally, a larger $a$ always increases the cost of honest operation and therefore tends to make deviation relatively more attractive. The magnitude of $a$ depends on the nature of the asset posted as stake and on its alternative uses. For instance, if staking requires locking highly liquid assets such as stablecoins that could otherwise be deployed in lending, market-making, or other yield-generating strategies, the opportunity cost may be substantial. By contrast, if the staked asset has limited alternative uses or is itself remunerated while locked, the effective opportunity cost is much lower. Hence, the effectiveness of staking as a disciplinary mechanism depends not only on the amount of collateral required, but also on the economic cost of immobilizing that collateral.

The strict margin $\xi>0$ allows the result to extend beyond the exact limit $\lambda=\rho=0$. By uniform continuity, the same qualitative ordering continues to hold for sufficiently small positive values of $\lambda$ and $\rho$. The proposition therefore identifies not only whether accumulated reputation acts predominantly as a source of discipline or as an additional opportunity for extraction, but also which economic and protocol primitives shift the system from one regime to the other.
\end{remark}

\begin{remark}\label{RemarkDeltaStake}
The previous discussion concerns the effect of staking on the monotonicity of
the temptation function with respect to reputational capital. In the exact
zero-update limit, the closed-form expression for $\tilde{\Delta}$ also allows us to
study the direct effect of $\kappa^S$ on the level of temptation itself.

Consider, for instance, the branch in which deviation is optimal at the reset
state. Under linear opportunity costs, Eq.~\eqref{eq:2608_1402} can be written
as
\[
(1-\beta)\tilde{\Delta}(F,M,0,0)
=
\Big[
\theta(1-\beta)-f
+
\kappa^S\big(a-(1-\beta)\big)
\Big]v(F,M)
+
\beta(\theta-\kappa^S)v_{\mathrm{reset}}
-\kappa^R.
\]
Differentiating with respect to the stake ratio gives
\[
(1-\beta)
\frac{\partial \tilde{\Delta}(F,M,0,0)}{\partial \kappa^S}
=
\big[a-(1-\beta)\big]v(F,M)
-
\beta v_{\mathrm{reset}}.
\]

This expression highlights a further trade-off generated by staking. When
$a<1-\beta$, both terms imply that increasing the stake reduces temptation,
so that staking has the conventional disciplinary effect throughout the state
space. When $a>1-\beta$, instead, the opportunity cost borne during honest
operation introduces an opposing force. In particular, for sufficiently large
reputation-induced activity,
\[
v(F,M)
>
\frac{\beta v_{\mathrm{reset}}}
{a-(1-\beta)},
\]
the derivative becomes positive and a larger stake increases the temptation
to deviate. However, this does not imply that the deviation region $\mathcal{D}$ expands if $\kappa^S$ increases, as shown in the following.

\end{remark}
\begin{remark}
The deviation region $\mathcal{D}$ never expands when the stake requirement $\kappa^S$ increases. That is, denoting the continuation region corresponding to the stake requirement $\kappa^S$ and $\lambda, \rho = (0,0)$ as $\mathcal{C}(\kappa^S)$, it holds that $\forall(F,M) \in \mathcal{C}(\kappa^S)$ and $\forall \tilde{\kappa}^S \ge \kappa^S$, $(F,M) \in \mathcal{C}(\tilde{\kappa}^S)$. We previously observed that this is true for $a \le 1 - \beta$, as $\tilde{\Delta}(F,M,0,0)$ is decreasing. Complementary, in the case $a > 1 - \beta$, we observe that $\tilde{\Delta}(F,M,0,0)$ is affine in $v(F,M)$: $\tilde{\Delta}(F,M,0,0) = A(\kappa^S) \, v(F,M) + C(\kappa^S)$. $v(F,M)$ coefficient is negative:
\begin{align*}
    A(\kappa^S) =& \theta \, (1 - \beta) - f + \kappa^S \, [a - (1 - \beta)] \le \\
    \le& \theta \, (1 - \beta) - f + \theta \, [a - (1 - \beta)] = \theta \, a - f \le a - f < 0
\end{align*}
When $C(\kappa^S)$ is negative, then $\tilde{\Delta}(F,M,0,0) < 0$. So, we assume $C(\kappa^S) > 0$. We can characterize $\mathcal{C}(\kappa^S) = \{ (F,M) \ | \ v(F,M) \ge - C(\kappa^S) / A(\kappa^S)\}$. The threshold $-C(\kappa^S) / A(\kappa^S)$ is decreasing in $\kappa^S$: $sign\big(\partial_{\kappa^S} (-C(\kappa^S) / A(\kappa^S)) \big) = sign\big( -\partial_{\kappa^S} C(\kappa^S) \, A(\kappa^S) + [a - (1 - \beta)] \, C(\kappa^S) \big)$. Two cases:
\begin{itemize}
    \item $C(\kappa^S) = \beta\, (\theta-\kappa^S) \, v_{reset}-\kappa^R$ In this case, $\partial_{\kappa^S} C(\kappa^S) = - \beta \, v_{reset}$ and
    \begin{align*}
        -\partial_{\kappa^S} C(\kappa^S)\,A(\kappa^S) + [a-(1-\beta)]\,C(\kappa^S) &= \beta\,v_{reset}\,\big\{\theta\,(1-\beta) - f + \kappa^S\,[a-(1-\beta)]\big\} \\
        &\quad + [a-(1-\beta)]\,\big[\beta\,(\theta-\kappa^S)\,v_{reset} - \kappa^R\big] \\
        &= \beta\,v_{reset}\,\big\{\theta\,(1-\beta) - f + \theta\,[a-(1-\beta)]\big\} - [a-(1-\beta)]\,\kappa^R \\
        &= \beta\,v_{reset}\,(\theta\,a - f) - [a-(1-\beta)]\,\kappa^R < 0
    \end{align*}
    where the sign follows from $\theta\,a \le a < f$, $a > 1-\beta$, and $\kappa^R \ge 0$.
    \item $C(\kappa^S) = \beta \, (f - a\,\kappa^S) \, v_{reset} - (1-\beta) \, \kappa^R$. This leads to $\partial_{\kappa^S} C(\kappa^S) = -a\,\beta\,v_{reset}$. Reiterating the same reasoning as above, we get:
    \begin{equation*}
        -\partial_{\kappa^S} C(\kappa^S)\,A(\kappa^S) + [a-(1-\beta)] \, C(\kappa^S) = (1-\beta) \, \big\{\beta\,v_{reset}\,(\theta\,a - f) - [a-(1-\beta)] \,\kappa^R\big\} < 0
    \end{equation*}
\end{itemize}
Thus, the threshold is decreasing in $\kappa^S$, and the sign of $\tilde{\Delta}$ cannot switch from negative to positive when increasing $\kappa^S$.
This non-expansion result relies on the maintained parameter restrictions of the model. In particular, under linear opportunity costs, $\kappa^S<\theta\le1$ and $a<f$ imply
\[
f-a\kappa^S>0
\]
throughout the admissible staking range. Hence, although Remark~\ref{RemarkDeltaStake} shows that a higher stake can locally increase temptation when $a>1-\beta$ and $v(F,M)$ is sufficiently large, the opportunity cost of staking is not allowed to become large enough to overturn the profitability of honest operation.

If this restriction were relaxed, a qualitatively different regime could emerge. In particular, sufficiently large staking requirements could make
\[
f-a\kappa^S\le0,
\]
so that the opportunity cost of collateral would fully offset, or even exceed, the per-unit fee earned under honest operation. In such an ``over-staking'' regime, the argument above no longer guarantees that $A(\kappa^S)<0$, and sufficiently high staking could potentially become self-defeating by making honest continuation less attractive and expanding the deviation region. We leave the analysis of this regime for future work.
\end{remark}
\begin{remark}
    In the case $\kappa^S=0$, the corresponding zero-update condition becomes
\[
\theta(1-\beta)-f \gtrless \pm\xi,
\]
or equivalently,
\[
\theta
\gtrless
\frac{f}{1-\beta} \pm\xi.
\]
That is, the monotonicity of $\tilde{\Delta}$ depends on whether the per-unit profit, at time $0$, of the opportunistic deviation is greater than the fee gained from the honest behavior, compounded over the infinite horizon.
\end{remark}

\section{Numerical Analysis}
\label{sec:numerical}

The previous sections provided an analytical characterization of the agent's
incentives, with particularly explicit results in the small-update regime.
These results distinguish three related but conceptually different objects.
First, the continuation and deviation regions are characterized through the
sign of the temptation function, with an explicit threshold representation in
the zero-update limit given by Corollary~\ref{eq:2608_1402} and
Eq.~\eqref{eq:0512_2045}. Second, Proposition~\ref{prop:0825_1626} and the
following remark characterize how temptation varies with reputational capital,
and how this dependence is affected by the staking requirement. Third,
Remark~\ref{RemarkDeltaStake} studies the direct effect of $\kappa^S$ on the
level of temptation at a fixed reputational state, while the subsequent remark
shows that these local effects do not imply an expansion of the deviation
region under the maintained parameter restrictions.

The numerical analysis is useful for three reasons. First, the closed-form
results derived in the previous section rely on the small-update approximation,
while the full Bellman problem remains well defined for arbitrary values of
the reputation-update parameters. Numerical solution therefore allows us to
assess how far the analytical predictions extend beyond the neighborhood in
which they are formally derived. Second, the theoretical results are mostly
qualitative, identifying the sign and direction of comparative statics,
whereas the numerical analysis quantifies their economic magnitude and shows
how large the continuation and deviation regions actually are. Third, several
of our results distinguish between local changes in the temptation function
and changes in the set of states in which deviation is optimal. Solving the
full model makes it possible to study these objects jointly and to verify
whether changes in the level or slope of temptation translate into economically
meaningful changes in the deviation region.

The objective of this section is therefore to investigate these mechanisms
numerically by solving the full Bellman problem. We proceed in three steps.
First, we study the aggregate policy implications of staking and of the
post-deviation identity mechanism by measuring the portion of the state space
in which deviation is optimal. Second, we examine how staking changes the
shape of the temptation function with respect to reputation and maturity,
providing a numerical counterpart to Proposition~\ref{prop:0825_1626}.
Finally, we assess the accuracy and robustness of the small-update
characterization by comparing its analytical predictions with the full
numerical solution as the update parameters $\lambda$ and $\rho$ increase.
The numerical analysis therefore complements the theory both by illustrating
the economic magnitude of the mechanisms identified analytically and by
showing how far the small-update results extend beyond their formal domain of
validity.

\subsection{Numerical Setup}

We solve the Bellman problem numerically by discretizing the state and control
spaces. Unless otherwise specified, the baseline environment is defined by the
following parameter choices:
\begin{enumerate}
    \item The deviation and continuation gross profit shares are
    $\theta = 0.7$ and $f=0.06$.
    \item The cost function is quadratic:
    $c(q) = \gamma q^2/2$ with $\gamma=0.25$.
    \item The feedback function is $s(q) = \sqrt{q}$.
    \item The opportunity cost is linear:
    $c^o(x) = a x$.
    \item The re-entering cost is $\kappa^R = 0.05$.
    \item The volume is Cobb--Douglas,
    $v(F,M) = F^{1/2}M^{1/2}$, while the maturity signal is
    $\sigma(M) = \sqrt{M}$.
    \item The post-deviation transitions are
    \[
    R^{new}(F,M)=(F_{\mathrm{reset}},M_{\mathrm{reset}})
    =(0.01,0.01),
    \]
    and
    \[
    R^{same}(F,M)
    =
    \big(
    (1-\lambda)F,
    (1-\rho)M+\rho\sigma(M)
    \big).
    \]
    \item The numerical values of $\kappa^S$, $\lambda$, $\rho$, $\beta$,
    and $a$ vary across experiments and are reported case by case.
\end{enumerate}

The state variable $F\in[0,1]$ is discretized into $N_F$ equally spaced points,
\[
F_j=\frac{j}{N_F-1},
\qquad
j=0,\ldots,N_F-1.
\]
Analogously,
\[
M_h=\frac{h}{N_M-1},
\qquad
h=0,\ldots,N_M-1,
\]
and the quality-control space is discretized as
\[
q_l=\frac{l}{N_q-1},
\qquad
l=0,\ldots,N_q-1.
\]
The state space is therefore represented by an $N_F\times N_M$ grid, while
the continuation branch is maximized over an $N_q$-point control grid.
Off-grid evaluations of the value function are obtained by linear
interpolation.

Starting from
\[
V_{(0)}=0\in\mathbb{R}_+^{N_F\times N_M},
\]
the deviation branch at iteration $k$ is
\begin{equation}
V^{dev}_{(k)}(F_j,M_h)
=
(\theta-\kappa^S)v(F_j,M_h)
-\kappa^R
+
\beta V_{(k-1)}\big(R(F_j,M_h)\big).
\end{equation}

The continuation branch is
\begin{equation}
V^{cont}_{(k)}(F_j,M_h)
=
f\,v(F_j,M_h)
-
c^o\left(\kappa^S v(F_j,M_h)\right)
+
\max_{l=0,\ldots,N_q-1}
\left[
-c(q_l)
+
\beta V_{(k-1)}
\big(T(F_j,M_h,q_l)\big)
\right].
\end{equation}

The value-function update is therefore
\begin{equation}
V_{(k)}(F_j,M_h)
=
\max
\left\{
V^{dev}_{(k)}(F_j,M_h),
V^{cont}_{(k)}(F_j,M_h)
\right\}.
\end{equation}

Iterations are continued until
\[
\left\|V_{(k)}-V_{(k-1)}\right\|_\infty
<
10^{-12}.
\]

\subsection{Aggregate Policy Effects and Identity Persistence}

We first study the aggregate consequences of the policy mechanisms introduced
in the theoretical analysis. To summarize the overall propensity to deviate,
we consider the Lebesgue measure of the deviation region,
\[
\mu(\mathcal{D}),
\]
over the normalized state space $\mathcal{X}=[0,1]^2$. Since
$\mu(\mathcal{X})=1$, this quantity can be interpreted directly as the fraction
of reputational states in which opportunistic deviation is optimal.

Figure~\ref{img:deviation_area} reports $\mu(\mathcal{D})$ as a function of
the required stake ratio $\kappa^S$. The figure serves two purposes. First, it
provides a fully endogenous numerical comparison between the identity-reset
mechanism $R^{new}$ and the same-identity mechanism $R^{same}$. Proposition
\ref{prop:R_temptation_comparison} showed that, holding the future-value
environment fixed, preserving a larger fraction of reputational capital after
misconduct makes deviation relatively more attractive. The numerical solution
confirms the same economic mechanism at the fixed point of the Bellman
problem: forcing the agent to restart from $R^{new}$ generally produces a
smaller deviation region than allowing part of the accumulated reputational
capital to survive under $R^{same}$.

Second, Figure~\ref{img:deviation_area} shows the aggregate effect of staking
and discounting. Increasing $\kappa^S$ generally reduces the size of the
deviation region, while a larger discount factor $\beta$ strengthens the value
of future activity and therefore reinforces reputational discipline. This
aggregate result is particularly useful when compared with the more local
comparative statics studied below. Indeed, the analytical results showed that
the effect of staking on the \emph{level} or the \emph{slope} of temptation
may become non-standard when the opportunity cost of staking is sufficiently
large. Figure~\ref{img:deviation_area} shows that such local incentive effects
need not translate mechanically into an expansion of the set of states in
which deviation is actually optimal.

The numerical grid used in these experiments is
$N_F=N_M=N_q=300$.

\begin{figure}[ht]
\centering
\includegraphics[width=\linewidth]{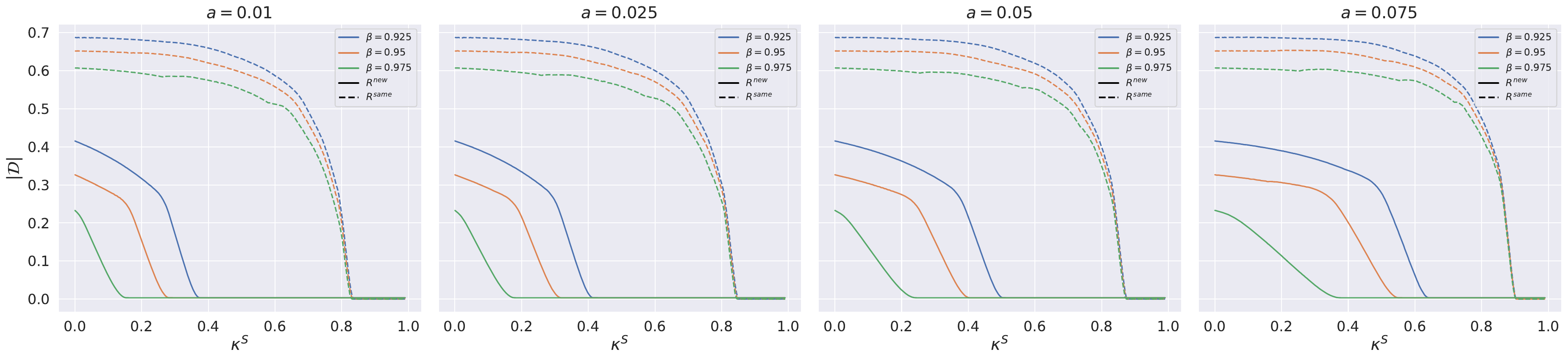}
\vspace{-.3cm}
\caption{Lebesgue measure of the deviation region $\mathcal{D}$ as a function
of the stake ratio $\kappa^S$. Continuous lines correspond to the
identity-reset mechanism $R^{new}$; dotted lines correspond to $R^{same}$.
Blue, orange, and green lines are for $\beta=0.925,0.95,0.975$.
Each panel reports a different opportunity-cost parameter,
$a=0.01,0.025,0.05,0.075$ from left to right.}
\label{img:deviation_area}
\end{figure}

\subsection{Staking and the Shape of Temptation}

We now move from the aggregate size of the deviation region to the shape of
the temptation function itself. Proposition~\ref{prop:0825_1626} showed that,
in the small-update regime, the dependence of $\Delta$ on reputational capital
is governed by
\[
(\theta-\kappa^S)(1-\beta)-f+a\kappa^S.
\]
In particular, the effect of increasing the stake on this coefficient changes
sign at
\[
a=1-\beta.
\]

Figure~\ref{img:delta_FM} investigates this mechanism numerically in the full
Bellman problem. We compute finite increments of $\Delta$ with respect to
$F$ and $M$,
\[
\Delta(F+\delta_F,M)-\Delta(F,M)
\]
and
\[
\Delta(F,M+\delta_M)-\Delta(F,M),
\]
respectively. Positive increments indicate that additional reputational
capital makes deviation relatively more attractive, while negative increments
indicate that reputation strengthens discipline.

The left panels consider $a<1-\beta$, whereas the right panels consider
$a>1-\beta$. The figure shows that the response of the slope of $\Delta$ to
staking changes qualitatively across these two regimes. When the opportunity
cost is relatively small, higher staking pushes the system toward a more
disciplining role of reputation. When the opportunity cost is sufficiently
large, instead, increasing the stake shifts the reputation dependence of
temptation in the opposite direction. This provides a numerical counterpart
to the comparative-static mechanism isolated in
Proposition~\ref{prop:0825_1626}.

As before, $N_F=N_M=N_q=300$.

\begin{figure}[ht]
\centering
\includegraphics[width=\linewidth]{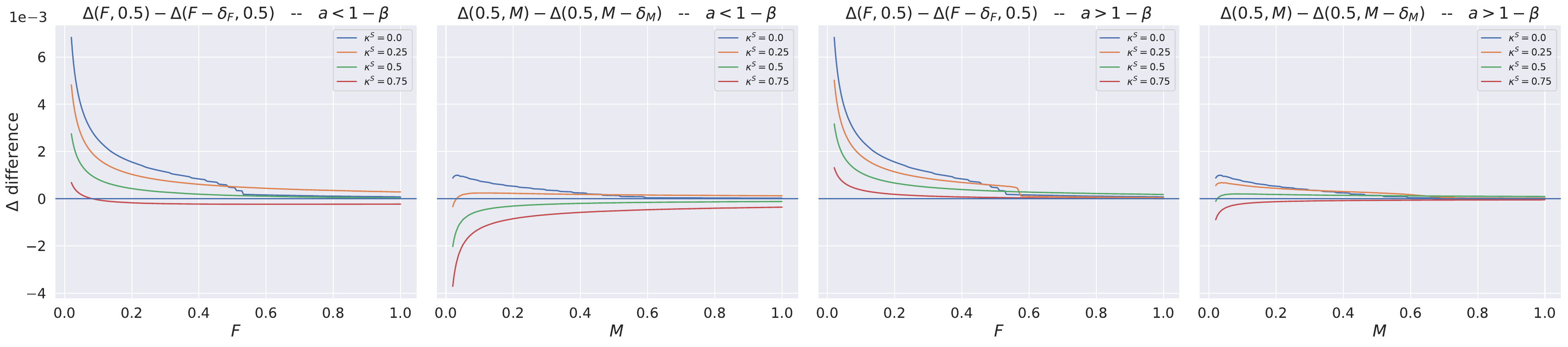}
\vspace{-.3cm}
\caption{$\Delta$ increments in $F$ (odd subplots) and $M$ (even subplots).
$\beta=0.95$, $f=0.1$, and $\theta=0.99$. The left subplots correspond to
$a=0.025<1-\beta$, while the right subplots correspond to
$a=0.075>1-\beta$. The horizontal line marks zero.
$\delta_F=\delta_M=3.3\cdot10^{-3}$.}
\label{img:delta_FM}
\end{figure}

\subsection{The Small-Update Benchmark}

We next assess directly the analytical characterization derived in
Section~\ref{sec:sur}. We begin from a parameter configuration close to the
zero-update limit, setting
\[
\lambda=\rho=0.001.
\]
In this regime, the analytical conditions derived from the
$\lambda=\rho=0$ problem should remain informative by uniform continuity.

Figure~\ref{img:small_update_regime} reports the full numerical temptation
function $\Delta(F,M)$ for increasing values of the stake ratio
$\kappa^S$. The blue contour identifies the numerical indifference region
\[
\mathcal{I}
=
\{(F,M):\Delta(F,M)=0\}.
\]

The figure illustrates two implications of the analytical results. First,
when the increasing condition in Eq.~\eqref{eq:0417_1344} is satisfied,
temptation rises with reputational capital and both continuation and deviation
states coexist. Second, sufficiently strong staking can move the economy into
the decreasing regime and, for the parameter configuration considered here,
eventually eliminate the deviation region altogether. Thus, close to the
zero-update limit, the numerical Bellman solution reproduces both the
monotonicity prediction and the threshold structure derived analytically.

The parameters are $\lambda=\rho=0.001$, $\beta=0.9$, and $a=0.03$.
The state and control grids are chosen sufficiently fine to make the numerical
boundary stable.

\begin{figure}[ht]
\centering
\includegraphics[width=\linewidth]{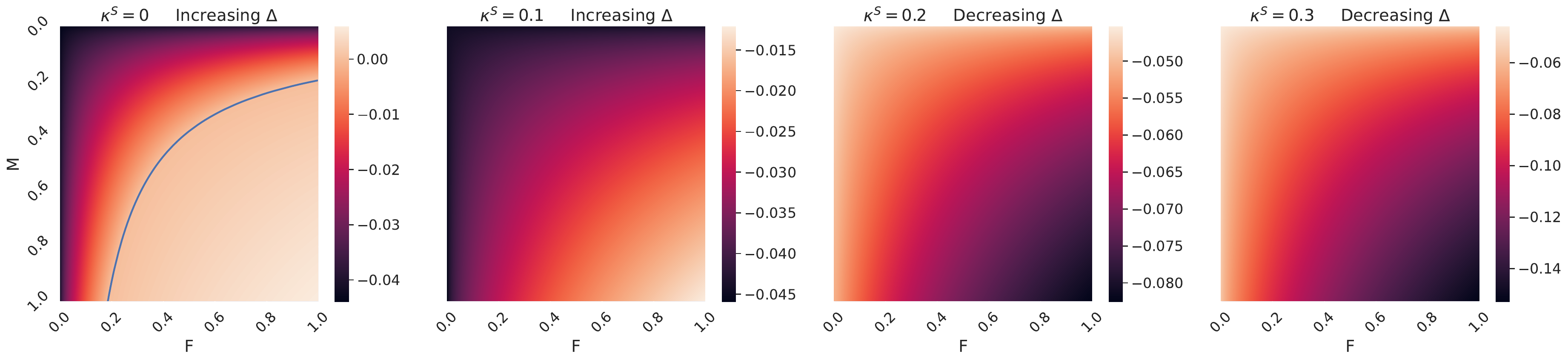}
\vspace{-.3cm}
\caption{Small-update regime with $\lambda=\rho=0.001$. The heatmaps report
$\Delta(F,M)$, while the blue contour identifies the numerical indifference
region $\mathcal{I}=\{(F,M):\Delta(F,M)=0\}$. From left to right,
$\kappa^S=0,0.1,0.2,0.3$.}
\label{img:small_update_regime}
\end{figure}

\subsection{Beyond the Small-Update Regime}

The previous figure shows that the analytical characterization is informative
when reputation updates are sufficiently slow. We now progressively increase
$\lambda$ and $\rho$ in order to identify where the zero-update approximation
starts to lose accuracy.

Figure~\ref{img:delta_heatmap_nosur} compares the numerical indifference
boundary with its analytical zero-update counterpart. The continuous contour
is obtained from the full Bellman solution, whereas the dotted contour is
computed from Eq.~\eqref{eq:2608_1402}. For very small update coefficients,
the two boundaries almost coincide. As $\lambda$ and $\rho$ increase, the
numerical boundary progressively departs from the analytical benchmark.
Moreover, the globally increasing structure predicted by the small-update
result begins to break down, first in specific portions of the state space.

\begin{figure}[ht]
\centering
\includegraphics[width=\linewidth]{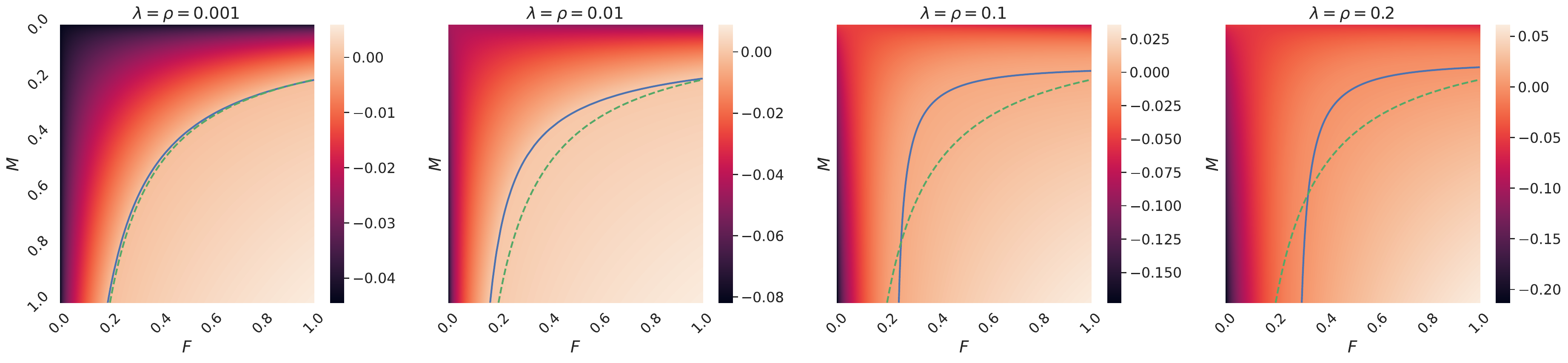}
\vspace{-.3cm}
\caption{$\Delta(F,M)$ heatmaps for increasing reputation-update rates. The
blue continuous line is the numerical indifference region
$\mathcal{I}=\Delta^{-1}(\{0\})$. The green dotted line is the analytical
indifference region in the zero-update limit, obtained from
Eq.~\eqref{eq:2608_1402}. From left to right,
$\lambda=\rho=0.001,0.01,0.1,0.2$.}
\label{img:delta_heatmap_nosur}
\end{figure}

To characterize this breakdown more systematically, we finally study the
monotonicity of $\Delta$ over the $(\lambda,\rho)$ plane. Define the local
finite-difference slopes
\begin{align}
\delta_F
&=
\frac{\Delta(F_2,M)-\Delta(F_1,M)}
{F_2-F_1},
&
\delta_M
&=
\frac{\Delta(F,M_2)-\Delta(F,M_1)}
{M_2-M_1},
\end{align}
and their extrema over the state grid,
\begin{align}
m_F &= \min_{F_1<F_2,M}\delta_F,
&
M_F &= \max_{F_1<F_2,M}\delta_F,
\\
m_M &= \min_{F,M_1<M_2}\delta_M,
&
M_M &= \max_{F,M_1<M_2}\delta_M.
\end{align}

For each update-parameter pair, we classify the temptation function as
globally increasing, globally decreasing, or non-monotone. For the $F$
dimension, define
\begin{equation}
\mathcal{M}_F
=
\big[1+\operatorname{sign}(m_FM_F)\big]
\frac{\operatorname{sign}(m_F)}{2}
=
\begin{cases}
-1,
&
m_F<M_F<0,
\\[0.1cm]
0,
&
m_F<0<M_F,
\\[0.1cm]
1,
&
0<m_F<M_F.
\end{cases}
\end{equation}
Thus,
\[
\mathcal{M}_F=-1
\]
corresponds to a globally decreasing temptation function in $F$,
\[
\mathcal{M}_F=1
\]
to a globally increasing one, and
\[
\mathcal{M}_F=0
\]
to a state-dependent, non-monotone response. The analogous indicator
$\mathcal{M}_M$ is defined using $m_M$ and $M_M$.

Figure~\ref{img:delta_monotonicity} reports these classifications over the
$(\lambda,\rho)$ plane. Close to the origin, the numerical solution reproduces
the monotonicity regime predicted by the small-update analysis. As the update
parameters increase, regions with mixed-sign increments emerge, showing where
the one-directional comparative statics cease to provide a global description
of the temptation function. The comparison between $\kappa^S=0$ and
$\kappa^S=0.3$ additionally shows how staking changes the extent of the region
over which the analytical monotonicity result remains informative.

\begin{figure}[ht]
\centering
\includegraphics[width=\linewidth]{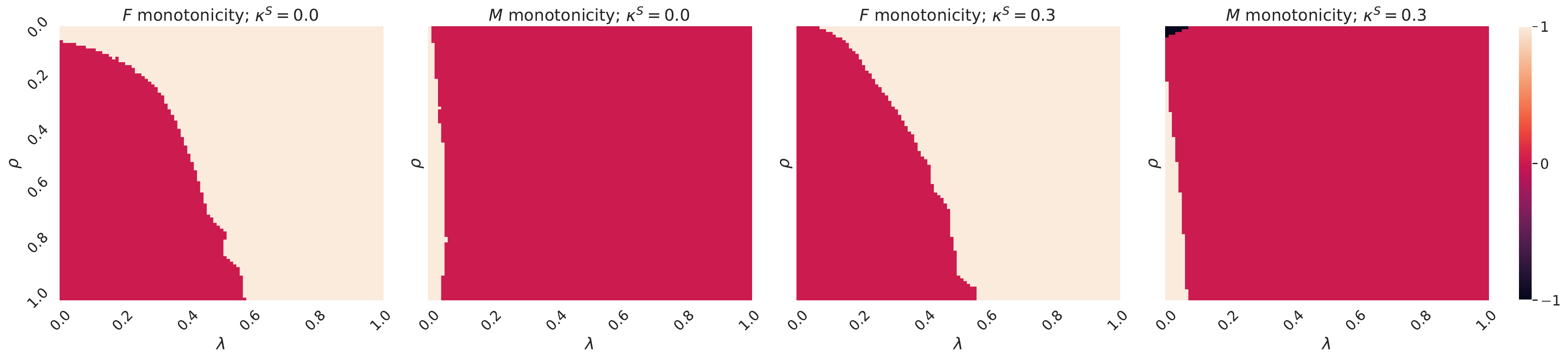}
\vspace{-.3cm}
\caption{Monotonicity of $\Delta$ in $F$ and $M$ as a function of the update
parameters $\lambda$ and $\rho$. A value $1$ denotes globally increasing
$\Delta$, $-1$ globally decreasing $\Delta$, and $0$ a non-monotone response.
The two left subplots correspond to $\kappa^S=0$ and the two right subplots to
$\kappa^S=0.3$. Odd and even subplots refer respectively to monotonicity in
$F$ and $M$. The remaining parameters are $\beta=0.9$ and $a=0.03$.}
\label{img:delta_monotonicity}
\end{figure}

\section{Optimal Effort and the Cost of Enforcement}
\label{sec:effort}

In this section, we analyze the impact of the state variables $F$ and $M$ and the stake coefficient $\kappa^S$ on the optimal effort provided by the agent. For tractability, since we are interested in the effort only when the agent does not deviate from honest behavior, we analyze the problem assuming $\mathcal{C} = \mathcal{X}$ -- that is, the agent always behaves honestly. This allows us to study a second dimension of the S\&S policy. On the one side, it helps preventing opportunistic deviation; on the other side, the staking mechanism may affect the effort optimally supplied by the agent.

The additional assumption we adopt in this section is:

\begin{assumption} \label{ass:honest}
    Let the continuation set equal the whole state space: $\mathcal{C} = \mathcal{X}$. Thus, $V=V^{cont}$ and it is the fixed point of the continuation-branch Bellman operator $\mathcal{B}^{cont} : W \in C_b(\mathcal{X}) \rightarrow V^{cont}_W \in C_b(\mathcal{X})$. Moreover, the maturity update $\sigma$ is independent of $F$; that is, $\sigma = \sigma(M)$ increasing and concave. Finally, the opportunity cost is linear: $c^o : v \rightarrow a \, v$ with $a < f$.
\end{assumption}

The following technical result will be useful to prove the monotonicity. The proof of the following proposition and the related corollary are in Appendix \ref{app_proofs_comparative}.

\begin{proposition}[$V$ concave] \label{prop:0509_0331}
Assume \ref{ass:update}--\ref{ass:cost}. Additionally, assume \ref{ass:honest}. Then, the value function $V(F, M)$ is concave in $F$.
\end{proposition}

\begin{corollary}[$U$ concave] \label{cor:0825_1542}
    The continuation kernel $U$ is strictly concave in $q$. Specifically, it admits an unique maximizer $q^* = \text{argmax}_q U(F, M, q)$.
\end{corollary}

In the following, we first find that the agent's response to the state variable is asymmetric: greater maturity increases effort, while a high reputation level decreases quality. Then, we explicitly derive that an increased staking amount, although necessary to reduce deviation temptation, degrades the service provided. Thus, it's in the protocol governance interest to set $\kappa^S$ as small as possible, provided it is enough to prevent deviation. For example, one can set $\kappa^S = \bar{\kappa}^S$, as discussed in Eq. \eqref{eq:0512_2047} or, in the small-update regime, one can use the condition derived by Eq. \eqref{eq:0512_2045}.

The key theoretical instruments used to describe the effort provided are the supermodularity results, see \cite{topkis1998supermodularity} for a detailed discussion. Specifically, a function $U(F,M)$ is said to be supermodular if, $\forall F_1 \le F_2$ and $\forall M_1 \le M_2$, the functional $\mathcal{M}$ assumes non-negative values:
\begin{equation}
    \mathcal{M}(U, F_{1,2}, M_{1,2}) = \mathcal{M} (U, F_1, F_2, M_1, M_2) = U(F_1, M_1) + U(F_2, M_2) - U(F_1, M_2) - U(F_2, M_1)
\end{equation}
Similarly, the function $U(F,M)$ is said to be submodular if, for every couple of pairs, it holds $\mathcal{M} (U, F_{1,2}, M_{1,2}) \le 0$. A well-established result \cite{topkis1978minimizing}  provides a sufficient condition for the optimal effort $q$ to be monotonically dependent on $F$ and $M$. If the function $U$ is submodular in $(F,q)$, then $q^* = argmax_q U(F, M, q)$ is non-increasing in $F$; similarly, if $U$ is supermodular in $(M, q)$, then $q^*(F,M) = argmax_{q} U(F, M, q)$ is non-decreasing in $M$. The supermodularity of $V$ is discussed in the following result, proved in Appendix \ref{app_proofs_comparative}.

\begin{proposition}[Optimal Effort - Comparative Statics] \label{prop:0512_2148}
Assume \ref{ass:update}--\ref{ass:cost}. Additionally, assume \ref{ass:honest}. Then, the optimal effort adopted by the agent, $q^*(F,M; \kappa^S) = argmax_{q \in [0,1]} U(F, M, q)$, is such that:
\begin{itemize}
    \item $q^*(F,M; \kappa^S)$ is non-increasing in $F$.
    \item $q^*(F,M; \kappa^S)$ is non-decreasing in $M$.
    \item $q^*(F,M; \kappa^S)$ is non-increasing in $\kappa^S$.
\end{itemize}
\end{proposition}

\begin{remark}
    The monotonicity of $q$ can be related to a kind of break-even condition. When the reputation $F$ increases, the profit from raising the effort to further increase the reputation is smaller due to the concavity of the volume function. Similarly, when the stake requirement $\kappa^S$ increases, there is a rise in the opportunity cost that makes higher volumes less profitable. Contrastingly, when the maturity $M$ increases, the agent finds it more profitable to raise the effort as higher future values of the reputation are better remunerated due to the positive mixed effect on the volume.
\end{remark}

\section{Conclusion}
\label{sec:conc}

This paper studies the economic role of reputation in environments where identity reset and reputation reconstruction are subject to limited frictions. Although autonomous AI-agent markets provide a particularly relevant application of this framework, the theoretical mechanism developed in this paper extends well beyond blockchain-based systems. Similar incentive problems arise whenever (i) reputation creates value by increasing the future economic activity that an agent can attract, and (ii) the same accumulated economic value can become the object of opportunistic extraction when identities can be abandoned and rebuilt at sufficiently low cost.

We formalize this trade-off through a dynamic model in which an autonomous agent repeatedly chooses between preserving reputational capital and opportunistically liquidating it. A first result is that reputation is indeed an economically valuable asset: a stronger and more mature reputational position increases the agent's value. At the same time, however, reputational capital has a dual role. By attracting larger economic flows, it increases not only the value of remaining in the market but also the amount of value that can be extracted through deviation. We derive sufficient conditions under which the latter effect dominates, so that temptation increases rather than decreases with accumulated reputation. Hence, stronger reputation does not necessarily translate into stronger discipline.

Identity persistence plays a central role in this mechanism. When misconduct forces the agent to restart from a new identity, a larger fraction of accumulated reputational capital is destroyed, reducing the continuation value available after deviation. By contrast, when the agent can continue operating under the same identity and preserve part of its reputation or maturity, deviation becomes relatively more attractive. This highlights that the effectiveness of a reputation system depends not only on how reputational information is accumulated, but also on how much of its economic value survives misconduct.

In the small-update regime, we obtain an explicit characterization of these incentives in terms of the primitive parameters of the model. The direction in which temptation changes with reputational capital depends on the relative strength of extractable value, honest fee income, discounting, and the opportunity cost generated by staking. In particular, staking has a non-trivial effect. A higher stake directly reduces the payoff obtainable through deviation, but at the same time increases the cost of honest operation. When the opportunity cost is sufficiently large, increasing the stake can therefore locally increase the temptation to deviate, even though, under the maintained baseline restrictions, the deviation region does not expand in the exact zero-update limit. This distinction between changes in the level of temptation and changes in the set of states in which deviation is actually optimal is also supported by the numerical analysis.

This result has a direct implication for protocol design. The effectiveness of staking depends not only on the amount of collateral required, but also on the economic characteristics of the asset posted as stake and on the mechanism through which staking is implemented. Collateral with valuable alternative uses may generate substantial opportunity costs and weaken the disciplinary effect of staking, whereas remunerated staking or assets with limited alternative uses can mitigate this distortion. Staking requirements should therefore be designed jointly with the choice of collateral and its remuneration scheme, rather than treated as a purely quantitative penalty.

The model also identifies a second enforcement trade-off concerning service quality. Conditional on honest continuation, the optimal effort is non-increasing in the staking requirement: stronger financial enforcement can therefore reduce the incentives of compliant agents to invest in quality. More generally, the comparative statics show that optimal effort decreases with the reputation level and increases with reputation maturity, reflecting the different marginal incentives generated by these two dimensions of reputational capital. Enforcement design should therefore not be evaluated solely according to its ability to deter deviation, but also according to the productive incentives it creates for agents that remain honest.

The numerical analysis complements these analytical results by solving the full Bellman problem outside the exact zero-update limit. It shows how staking and the post-deviation identity mechanism affect the size of the deviation region, illustrates how staking changes the dependence of temptation on reputation and maturity, and evaluates how far the small-update characterization remains informative as the reputation-update parameters increase. In particular, forcing the agent to restart from a new identity generally produces a smaller deviation region than allowing part of the accumulated reputational capital to survive after misconduct.

Overall, our results show that reputation alone need not be sufficient to sustain desirable behavior when economic identities are weakly persistent. Reputation is valuable precisely because it attracts future business, but this same mechanism also increases the amount of economic value that can be opportunistically exploited. Effective market design may therefore require reputational incentives to be complemented by mechanisms that make misconduct economically costly and identity replacement sufficiently consequential, while avoiding enforcement requirements that excessively distort honest participation.

Several directions remain open for future research. A natural extension is to incorporate stochastic feedback, heterogeneous users, endogenous reputation aggregation, and strategic evaluators. Another promising direction is to integrate the theoretical framework with empirical data from emerging autonomous-agent markets, studying how reputation evolves jointly with task allocation, payment flows, and identity dynamics. A further extension would be to explicitly analyze an over-staking regime in which collateral opportunity costs become sufficiently large to overturn the positive margin generated by honest operation. In such a regime, excessively demanding collateral requirements could potentially become self-defeating and expand the deviation region. More generally, as AI agents become increasingly autonomous economic actors, understanding how reputation, identity, and financial incentives interact will be essential for designing trustworthy and economically sustainable digital markets.

\section*{Acknowledgments}
FG acknowledges support from the research project ‘Dynamics and Information Research Institute—Quantum Information, Quantum Technologies’ within the agreement between UniCredit Bank and Scuola Normale Superiore.\\
FT acknowledges support from the grant PRIN2022 DD N. 104 of February 2, 2022 ”Liquidity and systemic risks in centralized and decentralized markets”, codice proposta 20227TCX5W - CUP J53D23004130006 funded by the European Union NextGenerationEU through the Piano Nazionale di Ripresa e Resilienza (PNRR).\\

\bibliographystyle{plainnat}
\bibliography{biblio}

\newpage
\appendix
\section{Symbols}
\label{app_symbols}
Table \ref{tab:symbols} summarizes the symbols notation used throughout the paper.
{ \rowcolors{2}{white}{black!5}
\begin{table}[H]
    \centering
    \begin{tabularx}{\textwidth}{L{5.8cm} | Y} 
        \toprule
        \textbf{Symbol} & \textbf{Meaning} \\
        \midrule
        $\mathcal{X}$ & State space: $[0,1]^2$ \\
        $(F,M)$ & Element in $\mathcal{X}$: reputation level, maturity of the reputation. Each is normalized in $[0,1]$. \\
        $h : \mathcal{X} \rightarrow \{0,1\}$ & Honesty control: $h=1$ means honest behavior; $h=0$ flags the opportunistic deviation. \\
        $q : \mathcal{X} \rightarrow \{0,1\}$ & Effort provided by the agent, normalized in $[0,1]$. \\
        $v:\mathcal{X} \rightarrow \mathbb{R}_+$ & Volume handled by the agent as a function of the state. \\
        $\pi:\mathcal{X} \times \{0,1\} \times [0,1] \rightarrow \mathbb{R}_+$ & Profit of the agent, as a function of the state and honesty and effort controls. \\
        $\pi^{cont}:\mathcal{X} \times [0,1] \rightarrow \mathbb{R}_+$ & Profit assuming the agent honestly behaves (function of state and effort). \\
        $\pi^{dev}:\mathcal{X} \rightarrow \mathbb{R}_+$ & Profit assuming the agent deviates (function of state only). \\
        $\theta \in (0,1]$ & Percentage of the volume appropriated by the agent when deviating. \\
        $\kappa^S \in [0,1]$ & Stake required by the protocol, expressed as a percentage of the volume handled by the agent. \\
        $\kappa^R \ge 0$ & Cost of starting a new identity. \\
        $f \in (0,1]$ & Percentage of the volume gained by the agent when honestly behaves (e.g., fee). \\
        $c^o: \mathbb{R}_+ \rightarrow \mathbb{R}_+$ & Opportunity cost of staking, as a function of the required stake. \\
        $c:[0,1] \rightarrow \mathbb{R}_+$ & Cost of providing a given level of effort. \\
        $a \in [0,f]$ & When $c^o$ is assumed to be linear, it is the opportunity cost coefficient: $c^o(x) = a \, x$. \\
        $\bar{\kappa}^S$ & Myopic staking thresholds: $\kappa^s \ge \bar{\kappa}^S$ is a sufficient condition for the agent to never deviate. \\
        $\sigma : \mathcal{X} \rightarrow [0,1]$ & Update of the maturity $M$, as a function of the current state . \\
        $T = (T_F, T_M): \mathcal{X} \times [0,1] \rightarrow \mathcal{X}$ & Continuation transition, function of the current state and the effort level. \\
        $(\lambda, \rho) \in [0,1]^2$ & Coefficients of the reputation and maturity components of the continuation transition. \\
        $R = (R_F, R_M): \mathcal{X} \rightarrow \mathcal{X}$ & Deviation transition function, function of the current state.\\
        $R^{same}$ and $R^{new}$ & Two examples of deviation transitions: same identity with penalized reputation and new identity, respectively. \\
        $\beta \in [0,1]$ & Time-discount factor. \\
        $V:\mathcal{X} \rightarrow \mathbb{R}_+$ & Value function of the agent. \\
        $V^{dev}_W:\mathcal{X} \rightarrow \mathbb{R}_+$ & Deviation branch with candidate state function $W$: $\pi^{dev}$ plus discounted $W$ evaluated in the deviation transition state. \\
        $V^{cont}_W:\mathcal{X} \rightarrow \mathbb{R}_+$ & Continuation branch with candidate state function $W$: $\pi^{cont}$ plus discounted $W$ evaluated in the continuation transition state (optimal effort). \\
        $U_W:\mathcal{X} \times [0,1] \rightarrow \mathbb{R}_+$ & Continuation kernel with candidate state function $W$: $-c(q)$ plus discounted $W$ evaluated in the continuation transition state (effort $q$). \\
        $\mathcal{C}, \mathcal{D}, \mathcal{I} \subset \mathcal{X}$ & Continuation, deviation, and indifference regions. \\
        $\Delta: \mathcal{X} \rightarrow [0,1]$ & Temptation function, defined as $V^{dev} - V^{cont}$. \\
        $L_M, L_M, L_v, L_\sigma, L_{R,F}, L_{R,MF}, L_{R,MM}$ & Lipschitz constants of different functions. \\
        \bottomrule
    \end{tabularx}
    \caption{Summary of the paper's main notation.}
    \label{tab:symbols}
\end{table}
}

\section{Proofs}
\label{app_proofs}
This appendix contains the proof of the results stated in the main text. The proofs are divided according to the main paper sections.

\subsection{Proofs of Section \ref{sec:dynamic}} \label{app_proofs_dynamic}
The following is the proof of Proposition \ref{prop:0509_0709}.
\begin{proposition*}[Existence, Monotonicity, and Sign]
Assume \ref{ass:update}--\ref{ass:cost}. Then there exists a unique, uniformly continuous value function $V$ satisfying Eq.~\eqref{eq:valuefunction_bellman}. Furthermore:
\begin{itemize}[label=$\triangleright$]
    \item $V$ is non-decreasing in $F$ and $M$;
    \item $V$ is non-increasing in $\kappa^S$;
    \item $V(F,M)\geq0$ for every $(F,M)\in\mathcal{X}$.
\end{itemize}
\end{proposition*}

\begin{proof}
First, we show existence and uniqueness. For every
$W\in C_b(\mathcal{X})$, continuity of $R$, $T$, $v$, $c^o$, and $c$,
together with compactness of the control set $[0,1]$, implies that
$\mathcal{B}W\in C_b(\mathcal{X})$. Thus, the Bellman operator maps
$C_b(\mathcal{X})$ into itself.

For any $W_1,W_2\in C_b(\mathcal{X})$, using
\[
\big|\max\{a,b\}-\max\{c,d\}\big|
\leq
\max\{|a-c|,|b-d|\},
\]
we obtain
\begin{align*}
\|\mathcal{B}W_1-\mathcal{B}W_2\|_\infty
&\leq
\max_{(F,M)\in\mathcal{X}}
\max\Bigg\{
\beta\,
\big|
W_1(R(F,M))-W_2(R(F,M))
\big|,
\\
&\qquad\qquad\qquad\qquad
\beta\,
\max_{q\in[0,1]}
\big|
W_1(T(F,M,q))-W_2(T(F,M,q))
\big|
\Bigg\}
\\
&\leq
\beta\,
\|W_1-W_2\|_\infty.
\end{align*}
Hence, $\mathcal{B}$ is a contraction with modulus $\beta$. The Banach
fixed-point theorem guarantees the existence of a unique fixed point
$V\in C_b(\mathcal{X})$. Since $\mathcal{X}$ is compact, $V$ is uniformly
continuous.

\medskip

$\bm{\triangleright}$ We now prove monotonicity in the reputational state.
Consider the set
\[
S
=
\left\{
W\in C_b(\mathcal{X})
\;\middle|\;
(F_1,M_1)\preceq(F_2,M_2)
\Longrightarrow
W(F_1,M_1)\leq W(F_2,M_2)
\right\},
\]
The set
$S$ is non-empty and closed under uniform convergence. We show that it is
invariant under $\mathcal{B}$.

Let $W\in S$ and
$(F_1,M_1)\preceq(F_2,M_2)$. For the deviation branch,
Assumption \ref{ass:volume} gives
\[
v(F_1,M_1)\leq v(F_2,M_2),
\]
while Assumption \ref{ass:deviation_transition} gives
\[
R(F_1,M_1)\preceq R(F_2,M_2).
\]
Since $W$ is increasing,
\[
W(R(F_1,M_1))
\leq
W(R(F_2,M_2)),
\]
and therefore
\[
V^{dev}_W(F_1,M_1)
\leq
V^{dev}_W(F_2,M_2).
\]

For the continuation branch, Assumption \ref{ass:update} implies, for every
$q\in[0,1]$,
\[
T(F_1,M_1,q)
\preceq
T(F_2,M_2,q),
\]
so that
\[
W(T(F_1,M_1,q))
\leq
W(T(F_2,M_2,q)).
\]
Hence,
\[
\max_{q\in[0,1]}U_W(F_1,M_1;q)
\leq
\max_{q\in[0,1]}U_W(F_2,M_2;q).
\]

Moreover, define
\[
g(x)=f\,x-c^o(\kappa^S x).
\]
By Assumption \ref{ass:cost},
\[
g'(x)
=
f-\kappa^S(c^o)'(\kappa^S x)>0,
\]
where we use $\kappa^S<\theta\leq1$ and $(c^o)'<f$.
Since $v$ is increasing,
\[
g(v(F_1,M_1))
\leq
g(v(F_2,M_2)).
\]
Thus,
\[
V^{cont}_W(F_1,M_1)
\leq
V^{cont}_W(F_2,M_2).
\]

Both branches preserve the componentwise order, and therefore
$\mathcal{B}W\in S$. Since $S$ is non-empty, closed, and invariant under
$\mathcal{B}$, the unique fixed point belongs to $S$. Hence, $V$ is
non-decreasing in both $F$ and $M$.

\medskip

$\bm{\triangleright}$ We next prove monotonicity in the staking requirement.
Let $0\leq\kappa^S_1\leq\kappa^S_2$ and denote by
$\mathcal{B}_{\kappa^S_i}$ the corresponding Bellman operators. For every
$W\in C_b(\mathcal{X})$, Assumption
\ref{ass:deviation_transition} implies that the post-deviation transition is
the same under the two staking levels. Hence,
\[
V^{dev}_{W,\kappa^S_1}(F,M)
\geq
V^{dev}_{W,\kappa^S_2}(F,M).
\]
Since $c^o$ is increasing,
\[
V^{cont}_{W,\kappa^S_1}(F,M)
\geq
V^{cont}_{W,\kappa^S_2}(F,M),
\]
and therefore
\[
\mathcal{B}_{\kappa^S_1}W
\geq
\mathcal{B}_{\kappa^S_2}W
\]
pointwise.

Moreover, each Bellman operator is monotone in its argument. Thus, if
$V_{\kappa^S_i}$ denotes its unique fixed point,
\[
V_{\kappa^S_2}
=
\mathcal{B}_{\kappa^S_2}V_{\kappa^S_2}
\leq
\mathcal{B}_{\kappa^S_1}V_{\kappa^S_2}.
\]
Iterating $\mathcal{B}_{\kappa^S_1}$ and taking the limit, using the
contraction property, gives
\[
V_{\kappa^S_2}
\leq
V_{\kappa^S_1}.
\]
Hence, $V$ is non-increasing in $\kappa^S$.

\medskip

$\bm{\triangleright}$ Finally, we prove non-negativity. Consider
\[
\widetilde{S}
=
\left\{
W\in C_b(\mathcal{X})
\;\middle|\;
W(F,M)\geq0,\quad\forall(F,M)\in\mathcal{X}
\right\}.
\]
The set $\widetilde{S}$ is non-empty and closed. Let
$W\in\widetilde{S}$. Since $q=0$ is feasible and $c(0)=0$,
\[
\max_{q\in[0,1]}U_W(F,M;q)
\geq
U_W(F,M;0)
=
\beta W(T(F,M,0))
\geq0.
\]
Furthermore, the function
\[
g(x)=f\,x-c^o(\kappa^S x)
\]
satisfies $g(0)=0$ and $g'(x)>0$, as shown above. Therefore
\[
f\,v(F,M)
-
c^o(\kappa^S v(F,M))
\geq0,
\]
and consequently
\[
V^{cont}_W(F,M)\geq0.
\]
Since
\[
(\mathcal{B}W)(F,M)
=
\max\left\{
V^{dev}_W(F,M),
V^{cont}_W(F,M)
\right\},
\]
we have
\[
(\mathcal{B}W)(F,M)\geq0.
\]
Thus, $\widetilde{S}$ is invariant under $\mathcal{B}$ and, by convergence
of the value iteration to the unique fixed point,
\[
V(F,M)\geq0
\qquad
\forall(F,M)\in\mathcal{X}.
\]
\end{proof}

The following is the proof of Proposition \ref{prop:R_temptation_comparison}.

\begin{proposition*}[Post-deviation comparison]
Let $W$ be non-decreasing in $F$ and $M$. Then, $\forall (F,M)$ such that:
\begin{equation}
    F \ge \frac{F_{reset} - \lambda \, s^{dev}}{1 - \lambda} \qquad\qquad M \ge \frac{M_{reset} - \rho \, \sigma(F,M)}{1 - \rho}
\end{equation}
we have:
\begin{equation}
\Delta_W^{same}(F,M)
-
\Delta_W^{new}(F,M)
=
\beta
\left[
W\big(R^{same}(F,M)\big)
-
W\big(R^{new}(F,M)\big)
\right]
\geq0.
\end{equation}
\end{proposition*}

\begin{proof}
The current deviation payoff and the continuation branch are identical under
the two mechanisms. The only difference is the continuation value following
misconduct. Therefore,
\begin{align*}
\Delta_W^{same}(F,M)
-
\Delta_W^{new}(F,M)
&=
\beta
\left[
W\big(R^{same}(F,M)\big)
-
W\big(R^{new}(F,M)\big)
\right].
\end{align*}
By the hypothesis in the statement,
\[
R^{new}(F,M)
\preceq
R^{same}(F,M),
\]
and since $W$ is non-decreasing in both components,
\[
W\big(R^{new}(F,M)\big)
\leq
W\big(R^{same}(F,M)\big).
\]
The result follows.
\end{proof}

Now, we show the proof of Proposition \ref{prop:lip}.
\begin{proposition*}[Lipschitz Continuity]
Assume \ref{ass:update}--\ref{ass:cost} and Assumption \ref{ass:R_lipschitz}. Assume furthermore that $v$ and $\sigma$ are Lipschitz continuous with constants $L_v$ and $L_\sigma$, respectively. Define $A_M := 1-\rho+\rho L_\sigma$. Suppose that
\begin{equation}
\beta A_M<1, \qquad\qquad \beta L_{R,MM}<1, \qquad\qquad \beta L_{R,F}<1.
\end{equation}

Then the value function $V$ is Lipschitz continuous separately in $M$ and
$F$. In particular, valid Lipschitz constants are
\begin{align}
L_M &= L_v \max \left\{ \frac{\theta-\kappa^S}{1-\beta \, L_{R,MM}}, \, \frac{f}{1-\beta \, A_M} \right\} \\
L_F &= \max \left\{ \frac{(\theta-\kappa^S) \, L_v + \beta \, L_{R,MF} \, L_M}{1-\beta \, L_{R,F}}, \, \frac{f \, L_v + \beta \, \rho \, L_\sigma \, L_M}{1-\beta \, (1-\lambda)} \right\} \nonumber
\end{align}
\end{proposition*}

\begin{proof}
The proof follows the same invariant-set argument used above, with the
additional contribution generated by the state dependence of the
post-deviation transition.

We first establish Lipschitz continuity with respect to $M$. Consider the set
\begin{align*}
S_M
=
\big\{
W\in C_b(\mathcal{X})
\ \big|\
|W(F,M_1)-W(F,M_2)|
\leq
L_M|M_1-M_2|,
\\
\forall F,M_1,M_2\in[0,1]
\big\}.
\end{align*}
The set $S_M$ is non-empty, since constant functions belong to it, and it is
closed under uniform convergence. We show that $S_M$ is invariant under the
Bellman operator.

Let $W\in S_M$. Using
\[
\big|
\max\{a,b\}
-
\max\{c,d\}
\big|
\leq
\max
\{
|a-c|,
|b-d|
\},
\]
we have
\begin{align*}
&
|(\mathcal{B}W)(F,M_1)-(\mathcal{B}W)(F,M_2)|
\\
&\leq
\max
\left\{
|V^{dev}_W(F,M_1)-V^{dev}_W(F,M_2)|,
\,
|V^{cont}_W(F,M_1)-V^{cont}_W(F,M_2)|
\right\}.
\end{align*}

Consider first the deviation branch. Since $R_F$ is independent of $M$,
\begin{align*}
&
|V^{dev}_W(F,M_1)-V^{dev}_W(F,M_2)|
\\
&\leq
(\theta-\kappa^S)
|v(F,M_1)-v(F,M_2)|
\\
&\qquad
+
\beta
\left|
W\big(R_F(F),R_M(F,M_1)\big)
-
W\big(R_F(F),R_M(F,M_2)\big)
\right|
\\
&\leq
\left[
(\theta-\kappa^S)L_v
+
\beta L_M L_{R,MM}
\right]
|M_1-M_2|.
\end{align*}

For the continuation branch, define as before
\[
g(x)=f x-c^o(\kappa^S x).
\]
From Assumption \ref{ass:cost},
\[
0\leq g'(x)\leq f,
\]
and therefore
\[
|g(v(F,M_1))-g(v(F,M_2))|
\leq
fL_v|M_1-M_2|.
\]

Moreover, $T_F(F,q)$ does not depend on $M$, while
\begin{align*}
|T_M(F,M_1)-T_M(F,M_2)|
&=
\left|
(1-\rho)(M_1-M_2)
+
\rho
\big[
\sigma(F,M_1)-\sigma(F,M_2)
\big]
\right|
\\
&\leq
(1-\rho+\rho L_\sigma)|M_1-M_2|
\\
&=
A_M|M_1-M_2|.
\end{align*}
Hence,
\begin{align*}
&
\left|
\max_{q\in[0,1]}U_W(F,M_1;q)
-
\max_{q\in[0,1]}U_W(F,M_2;q)
\right|
\\
&\leq
\beta L_M A_M |M_1-M_2|.
\end{align*}

It follows that
\begin{align*}
&
|V^{cont}_W(F,M_1)-V^{cont}_W(F,M_2)|
\\
&\leq
\left[
fL_v+\beta L_M A_M
\right]
|M_1-M_2|.
\end{align*}

Therefore, $S_M$ is invariant provided that
\[
(\theta-\kappa^S)L_v
+
\beta L_M L_{R,MM}
\leq
L_M
\]
and
\[
fL_v+\beta L_M A_M
\leq
L_M.
\]
Under the conditions
\[
\beta L_{R,MM}<1
\qquad\text{and}\qquad
\beta A_M<1,
\]
these inequalities are satisfied by
\[
L_M
=
L_v
\max
\left\{
\frac{\theta-\kappa^S}
{1-\beta L_{R,MM}},
\,
\frac{f}
{1-\beta A_M}
\right\}.
\]
The unique fixed point therefore belongs to $S_M$, establishing Lipschitz
continuity in $M$.

We now turn to $F$. Consider the set
\begin{align*}
S_F
=
\big\{
W\in C_b(\mathcal{X})
\ \big|\
&|W(F_1,M)-W(F_2,M)|
\leq
L_F|F_1-F_2|,
\\
&
|W(F,M_1)-W(F,M_2)|
\leq
L_M|M_1-M_2|
\big\}.
\end{align*}
As before, this set is non-empty and closed. We show that it is invariant.

For the deviation branch,
\begin{align*}
&
|V^{dev}_W(F_1,M)-V^{dev}_W(F_2,M)|
\\
&\leq
(\theta-\kappa^S)
|v(F_1,M)-v(F_2,M)|
\\
&\quad
+
\beta
\left|
W\big(R_F(F_1),R_M(F_1,M)\big)
-
W\big(R_F(F_2),R_M(F_2,M)\big)
\right|.
\end{align*}
Adding and subtracting
$W(R_F(F_2),R_M(F_1,M))$ gives
\begin{align*}
&
|V^{dev}_W(F_1,M)-V^{dev}_W(F_2,M)|
\\
&\leq
\Big[
(\theta-\kappa^S)L_v
+
\beta L_F L_{R,F}
+
\beta L_M L_{R,MF}
\Big]
|F_1-F_2|.
\end{align*}

For the continuation branch,
\[
|T_F(F_1,q)-T_F(F_2,q)|
=
(1-\lambda)|F_1-F_2|,
\]
while
\[
|T_M(F_1,M)-T_M(F_2,M)|
\leq
\rho L_\sigma |F_1-F_2|.
\]
Hence,
\begin{align*}
&
|V^{cont}_W(F_1,M)-V^{cont}_W(F_2,M)|
\\
&\leq
\Big[
fL_v
+
\beta(1-\lambda)L_F
+
\beta\rho L_\sigma L_M
\Big]
|F_1-F_2|.
\end{align*}

It is therefore sufficient that
\[
(\theta-\kappa^S)L_v
+
\beta L_F L_{R,F}
+
\beta L_M L_{R,MF}
\leq
L_F
\]
and
\[
fL_v
+
\beta(1-\lambda)L_F
+
\beta\rho L_\sigma L_M
\leq
L_F.
\]
Under $\beta L_{R,F}<1$, these inequalities are satisfied by
\begin{align*}
L_F
=
\max
\left\{
\frac{
(\theta-\kappa^S)L_v
+
\beta L_{R,MF}L_M
}
{1-\beta L_{R,F}},
\,
\frac{
fL_v
+
\beta\rho L_\sigma L_M
}
{1-\beta(1-\lambda)}
\right\}.
\end{align*}

Thus, $S_F$ is invariant under the Bellman operator, and its unique fixed
point belongs to $S_F$. The value function is therefore Lipschitz continuous
in both $F$ and $M$ with the constants stated in
the statement.
\end{proof}

Next, we report the proof of Proposition \ref{prop:temptation_monotonicity}.

\begin{proposition*}[General Monotonicity]
Assume \ref{ass:update}--\ref{ass:cost} and the additional assumptions of
Proposition \ref{prop:lip}. Then, for every $\bar{F},\bar{M}>0$, if
\begin{equation}
\begin{cases} \big[ \theta-f-\kappa^S \big] \, \partial_F v(1,\bar{M}) \geq \beta \, L_F \, (1-\lambda) + \beta \, \rho \, L_M \, L_\sigma & \mathbf{(F)} \\[0.3cm]
\big[ \theta-f-\kappa^S \big] \, \partial_M v(\bar{F},1) \geq \beta \, L_M \, (1-\rho) + \beta \, \rho \, L_M \, L_\sigma & \mathbf{(M)} \end{cases}
\end{equation}
then the temptation function $\Delta(F,M)$ is non-decreasing in the corresponding variable on $[\bar{F},1]\times[\bar{M},1]$.
\end{proposition*}

\begin{proof}
We prove the result for $M$; the argument for $F$ is analogous. Let $M_1>M_2\geq\bar{M}$. From Eq.~\eqref{eq:deviation_branch},
\begin{equation*}
V^{dev}(F,M_1)-V^{dev}(F,M_2) =  (\theta-\kappa^S) \, \big[ v(F,M_1)-v(F,M_2) \big] + \beta \, \left[ V\big(R(F,M_1)\big) - V\big(R(F,M_2)\big) \right]
\end{equation*}

By Assumption \ref{ass:deviation_transition},
\[
R(F,M_2)
\preceq
R(F,M_1),
\]
and Proposition \ref{prop:0509_0709} establishes that $V$ is
non-decreasing in the reputational state. Therefore,
\[
V\big(R(F,M_1)\big)
-
V\big(R(F,M_2)\big)
\geq0.
\]
Hence,
\begin{equation}
\label{eq:dev_lower_bound_M}
V^{dev}(F,M_1)-V^{dev}(F,M_2)
\geq
(\theta-\kappa^S)
\big[
v(F,M_1)-v(F,M_2)
\big].
\end{equation}

For the continuation branch, let
\[
q_1
\in
\argmax_{q\in[0,1]}
U(F,M_1;q).
\]
Using the fact that the maximizers at $M_1$ and $M_2$ need not coincide, we
obtain
\begin{align*}
&
V^{cont}(F,M_1)-V^{cont}(F,M_2)
\\
\leq{}&
f
\big[
v(F,M_1)-v(F,M_2)
\big]
\\
&+
\beta
\left[
V\big(T(F,M_1,q_1)\big)
-
V\big(T(F,M_2,q_1)\big)
\right].
\end{align*}
The opportunity-cost term has been omitted from this upper bound because
$c^o$ is increasing and
$v(F,M_1)\geq v(F,M_2)$.

Since $T_F$ is independent of $M$, Proposition \ref{prop:lip} implies
\begin{align*}
&
V\big(T(F,M_1,q_1)\big)
-
V\big(T(F,M_2,q_1)\big)
\\
\leq{}&
L_M
\left[
(1-\rho)(M_1-M_2)
+
\rho
\big(
\sigma(F,M_1)-\sigma(F,M_2)
\big)
\right].
\end{align*}
Therefore,
\begin{align}
&
V^{cont}(F,M_1)-V^{cont}(F,M_2)
\nonumber\\
\leq{}&
f
\big[
v(F,M_1)-v(F,M_2)
\big]
+
\beta L_M(1-\rho)(M_1-M_2)
\nonumber\\
&+
\beta\rho L_M
\big[
\sigma(F,M_1)-\sigma(F,M_2)
\big].
\label{eq:cont_upper_bound_M}
\end{align}

Combining
Eqs.~\eqref{eq:dev_lower_bound_M} and
\eqref{eq:cont_upper_bound_M} gives
\begin{align*}
&
\Delta(F,M_1)-\Delta(F,M_2)
\\
\geq{}&
\big[
\theta-f-\kappa^S
\big]
\big[
v(F,M_1)-v(F,M_2)
\big]
\\
&-
\beta L_M(1-\rho)(M_1-M_2)
\\
&-
\beta\rho L_M
\big[
\sigma(F,M_1)-\sigma(F,M_2)
\big].
\end{align*}

Dividing by $M_1-M_2>0$ and using the mean-value theorem,
coordinatewise concavity of $v$, and the Lipschitz property of $\sigma$, we
obtain
\[
\frac{
v(F,M_1)-v(F,M_2)
}{
M_1-M_2
}
\geq
\partial_M v(\bar{F},1)
\]
and
\[
\frac{
\sigma(F,M_1)-\sigma(F,M_2)
}{
M_1-M_2
}
\leq
L_\sigma.
\]
It follows that
\begin{align*}
\frac{
\Delta(F,M_1)-\Delta(F,M_2)
}{
M_1-M_2
}
\geq{}&
\big[
\theta-f-\kappa^S
\big]
\partial_M v(\bar{F},1)
\\
&-
\beta L_M(1-\rho)
-
\beta\rho L_ML_\sigma.
\end{align*}
Condition $\mathbf{(M)}$ in the statement therefore implies
\[
\Delta(F,M_1)
\geq
\Delta(F,M_2).
\]

The proof for $F$ proceeds analogously. In particular, the
order-preserving property of $R$ implies that the post-deviation continuation
term is again non-negative and can therefore be omitted when constructing a
sufficient lower bound for the change in the temptation function. This yields
condition $\mathbf{(F)}$ in the statement.
\end{proof}

\subsection{Proofs of Section \ref{sec:sur}} \label{app_proofs_sur}

The following is the proof of Proposition \ref{prop:0825_1616}.
\begin{proposition*}[Explicit form of the value functions]
    Assume \ref{ass:update} to \ref{ass:volume}, and $\lambda = \rho = 0$. Then, we define $\tilde{V}^{dev}$ and $\tilde{V}^{cont}$ as the value of the branches assuming that, at time $0$, the agent chooses to deviate or to honestly behave. $\tilde{V}^{dev}$ and $\tilde{V}^{cont}$ read:
    \begin{align}
    \tilde{V}^{dev}(F,M,0,0) &= (\theta - \kappa^S) \, v(F, M) + \frac{1}{1-\beta} \, \max \begin{cases} (\theta - \kappa^S) \, \beta v_{reset} - \kappa^R \\ \beta \big[ f \, v_{reset} - c^o(\kappa^S \, v_{reset}) \big] - (1 - \beta) \, \kappa^R \end{cases} \\
    \tilde{V}^{cont}(F,M,0,0) &= \frac{1}{1-\beta} \, \bigg[ fv(F,M) - c^o(\kappa^S \, v(F,M)) \bigg] \nonumber
    \end{align}
\end{proposition*}

\begin{proof}
    First, observe that, for a generic point $(F,M) \in \mathcal{X}$, if the agent chooses to honestly behave once, it will always honestly behave, as we observed that $T(F,M, q) = (F,M)$ -- that is, if $h(F,M,0,0)=1$, then the state $(F,M)$ is absorbing. Then, we compute the value function in $(F_{reset}, M_{reset}, 0, 0)$, which is absorbing independently of the agent's choice. In this case, the future discounted value is the same in the continuation and deviation branch (it is $\beta \, V(F_{reset}, M_{reset}, 0, 0)$). Thus, to compare the profits is enough to look at the instantaneous rewards. Specifically:
    \begin{equation*}
        V_{reset}(0, 0) = \frac{1}{1-\beta} \, \begin{cases} (\theta - \kappa^S) \, v_{reset} - \kappa^R & if \ (\theta - \kappa^S - f) \, v_{reset} \ge \kappa^R - c^o(\kappa^S \, v_{reset}) \\ f \, v_{reset} - c^o(\kappa^S \, v_{reset}) & otherwise\end{cases}
    \end{equation*}
    Lastly, for a generic point $(F,M) \in \mathcal{X}$ we can rewrite the continuation branch as:
    \begin{align*}
    \tilde{V}^{cont}(F,M, 0, 0) &= f \, v(F,M) - c^o(\kappa^S \, v(F,M)) + \beta \, V^{cont}(F,M, 0, 0) = \\
    &= \frac{1}{1-\beta} \, \big\{ f \, v(F,M) - c^o(\kappa^S \, v(F,M)) \big\}
    \end{align*}
    while for the deviation branch we just replace $V_{reset}(0, 0)$ with the previous expression. Putting the pieces together, we obtain the claim.
\end{proof}

Next, we show the proof of Proposition \ref{prop:0825_1626}.

\begin{proposition*}[Small-update regime monotonicity]
    Assume \ref{ass:update} to \ref{ass:volume}. Assume that one of the following conditions hold:
    \begin{align}
        (\theta - \kappa^S) \, (1-\beta) - f \ge \xi & \qquad\qquad\text{\textbf{Increasing}} \\
        (\theta - \kappa^S) \, (1-\beta) - f + \kappa^S \, \le - \xi & \qquad\qquad\text{\textbf{Decreasing}} \nonumber
    \end{align}
    With linear opportunity cost $c^o(v) = a \, v$, these conditions can be sharpened as:
    \begin{equation}
        (\theta - \kappa^S) \, (1-\beta) - f + a \, \kappa^S \gtrless \pm \xi
    \end{equation}
    Then, $\forall \bar{\delta}, \bar{F}, \bar{M} \in (0,1)$, there exists a small-update regime $[0, \bar{\lambda}) \times [0, \bar{\rho})$ such that $\forall F_2 \in [\bar{F},1]$, $\forall M_2 \in [\bar{M},1]$, $\forall F_1 \in [0, F_2 \, \bar{\delta}]$, and $\forall M_1 \in [0, M_2 \, \bar{\delta}]$, it holds:
    \begin{align}
        \tilde{\Delta}(F_1,M_2, \lambda, \rho) \le \tilde{\Delta}(F_2,M_2, \lambda, \rho) & \qquad\qquad \tilde{\Delta}(F_2,M_1, \lambda, \rho) \le \tilde{\Delta}(F_2,M_2, \lambda, \rho) & \qquad\qquad\text{\textbf{Increasing}} \\
        \tilde{\Delta}(F_1,M_2, \lambda, \rho) \ge \tilde{\Delta}(F_2,M_2, \lambda, \rho) & \qquad\qquad \tilde{\Delta}(F_2,M_1, \lambda, \rho) \ge \tilde{\Delta}(F_2,M_2, \lambda, \rho) & \qquad\qquad\text{\textbf{Decreasing}} \nonumber
    \end{align}
\end{proposition*}

\begin{proof}    
    Let us consider the monotonicity in $M$ and consider $F, M_1, M_2 \in [0,1]$, with $M_1 \le M_2 \, \bar{\delta}$; the monotonicity in $F$ is symmetric. Let rewrite $M_1$ as $\delta \, M_2$ with $\delta = M_1 / M_2 \le \bar{\delta}$. The proof scheme is as follows. To show the monotonicity, we study the sign of the function
    \begin{equation*}
        \varphi(F, M, \lambda, \rho, \delta) = \frac{\tilde{\Delta}(F, M, \lambda, \rho) - \tilde{\Delta}(F, \delta \, M, \lambda, \rho)}{(1 - \delta) \, M} \gtreqless 0
    \end{equation*}
    defined in the domain $[\bar{F},1] \times [\bar{M}, 1] \times [0,1]^2 \times [0, \bar{\delta}]$. Specifically, we study $\varphi(F, M, \lambda, \rho, \delta) $. Furthermore, $\varphi$ is (uniformly) continuous, as it is the composition of continuous functions. $\varphi(\cdot, \cdot, 0, 0, \cdot)$ can be rewritten as:
    \begin{align*}
        \varphi(F, M, 0, 0, \delta) &= \frac{\big[ (\theta - \kappa^S) \, (1 - \beta) - f \big] \, \big[ v(F, M_2) - v(F, M_1) \big] + c^o\big( \kappa^S \, v(F, M_2) \big) - c^o\big( \kappa^S \, v(F, M_1) \big)}{(1 - \beta) \, (1 - \delta) \, M} = \\
        &= \frac{\big[ (\theta - \kappa^S) \, (1 - \beta) - f \big] \, \partial_M v(F, \tilde{M}) \, (1 - \delta) \, M + c^o\big( \kappa^S \, v(F, M_2) \big) - c^o\big( \kappa^S \, v(F, M_1) \big)}{(1 - \beta) \, (1 - \delta) \, M}
    \end{align*}
    where we used the mean-value theorem and $\tilde{M} \in (M_1, M_2)$. According to the case in \eqref{eq:0417_1344}, we can lower bound the opportunity cost component with $0$ ("Increasing") or upper bound with $\kappa^S \, \partial_M v(F, \tilde{M}) \, (1 - \delta) \, M$ ("Decreasing"). Keeping in mind that $\partial_M v(F, \tilde{M}) \ge \partial_M v(\bar{F}, 1)$ due to the sign of the second-order derivatives of $v$, we have:
    \begin{align*}
        \varphi(F, M, 0, 0, \delta) \ge \frac{\xi \, \partial_M v(\bar{F}, 1)}{1 - \beta} > 0 & \qquad\qquad\text{\textbf{Increasing}} \\
        \varphi(F, M, 0, 0, \delta) \le - \frac{\xi \, \partial_M v(\bar{F}, 1)}{1 - \beta} < 0 & \qquad\qquad\text{\textbf{Decreasing}} \nonumber
    \end{align*}
    By the uniform continuity, $\exists$ a small-update regime $[0, \bar{\lambda}) \times [0, \bar{\rho})$ such that $\varphi(F, M, \lambda, \rho, \delta) > 0$. This, in turns, implies $\tilde{\Delta}(F, M_2, \lambda, \rho) > \tilde{\Delta}(F, M_1, \lambda, \rho)$ (in the increasing case) or $\varphi(F, M, \lambda, \rho, \delta) < 0 \implies \tilde{\Delta}(F, M_2, \lambda, \rho) < \tilde{\Delta}(F, M_1, \lambda, \rho)$ (in the decreasing case), thus proving the statament.
\end{proof}

\subsection{Proofs of Section \ref{sec:effort}} \label{app_proofs_comparative}

The following is the proof of Proposition \ref{prop:0509_0331}.

\begin{proposition*}[$V$ concave]
Assume \ref{ass:update}--\ref{ass:cost}. Additionally, assume \ref{ass:honest}. Then, the value function $V(F, M)$ is concave in $F$.
\end{proposition*}

\begin{proof}
We prove this statement via a Banach argument. Let $S$ be the set defined as:
\begin{align*}
    S = \big\{ W \in C_b (\mathcal{X}) \ | & \ \forall F_1 \le F_2, \forall M, \forall \xi \in [0,1], \ \mathbf{(1)} \ W(F_1, M) \le W(F_2, M) \text{ and }\\
    & \mathbf{(2)} \ W\big( \xi \, F_1 + (1-\xi) \, F_2, M \big) \ge \xi \, W(F_1, M) + (1-\xi) \, W(F_2, M) \big\}
\end{align*}
Condition $\mathbf{(1)}$ says that we are considering functions that are non-decreasing in $F$; condition $\mathbf{(2)}$ embodies the concavity in $F$. We want to show that $V \in S$ exploiting Banach's fixed point theorem. $S \neq \emptyset$ as the constant functions are in $S$. Furthermore, it is closed as we are working with the uniform convergence, thus both the monotonicity in $F$ and the concavity are preserved. As for the invariance, let us consider $W \in S$. The task is to show that $\mathcal{B}^{cont} W \in S$ -- that is, $\mathcal{B}^{cont} W$ satisfies $\mathbf{(1)}$ and $\mathbf{(2)}$. The former has already been proven in Proposition \ref{prop:0509_0709}. As for the concavity, let us consider
\begin{equation*}
    \mathcal{B}^{cont} W\big( \xi \, F_1 + (1-\xi) \, F_2, M \big) = (f - a \, \kappa^S ) \, v\big( \xi \, F_1 + (1-\xi) \, F_2, M \big) + \max_{q \in [0,1]} U_W\big( \xi \, F_1 + (1-\xi) \, F_2, M, q \big)
\end{equation*}
The first term of the sum is concave by Assumption \ref{ass:volume}. As for the second term, from traditional calculus, we know that the maximum function, $\max_q U_W$, is concave in $F$ if the argument $U_W$ is jointly concave in $(F,q)$\footnote{Let's consider $F_1, F_2$, and $\xi \in [0,1]$. Let $q_j \in argmax_q U_W(F_j, M, q)$, with $j=1,2$. Due the convexity of the $q$ domain $[0,1]$, we have that $\xi \, q_1 + (1-\xi) \, q_2$. Thus: $\max_q U_W(\xi \, F_1 + (1-\xi) \, F_2, M, q) \ge U_W(\xi \, F_1 + (1-\xi) \, F_2, M, \xi \, q_1 + (1-\xi) \, q_2) \ge \xi \, U_W(F_1, q_1) + (1 - \xi) \, U_W(F_2, q_2) = \xi \, \max_q U_W(F_1, q) + (1 - \xi) \, \max_q U_W(F_2, q)$.}. So, we rewrite:
\begin{align*}
    U_W\big(\xi \, F_1 + (1-\xi) \, F_2,& M, \xi \, q_1 + (1-\xi) \, q_2\big) = -c\big(\xi \, q_1 + (1-\xi) \, q_2 \big) + \\
    & + \beta \, W \bigg( (1-\lambda) \, \big[ \xi \, F_1 + (1-\xi) \, F_2 \big] + \lambda \, s\big( \xi \, q_1 + (1-\xi) \, q_2 \big), T_M(M) \bigg) \ge \\
    & \ge - \xi \, c(q_1) - (1-\xi) \, c(q_2) + \\
    & + \beta \, W \bigg( (1-\lambda) \, \big[ \xi \, F_1 + (1-\xi) \, F_2 \big] + \lambda \, \xi \, s(q_1) + \lambda \, (1-\xi) \, s(q_2), T_M(M) \bigg)
\end{align*}
The cost function $c$ is convex, so $-c$ is concave; the $W$ contribution is bounded by the invariance hypothesis $W \in S$, guaranteeing both $W$ monotone and concave. This leads to the joint concavity. To summarize, we have that $\mathcal{B}^{cont} W$ is the sum of two concave pieces, so it is concave and belongs to $S$. Thus, $S$ is invariant for $\mathcal{B}^{cont}$ and the fixed point of $\mathcal{B}^{cont}$, that is $V$, is concave in $F$.
\end{proof}

Next, there is the proof of Corollary \ref{cor:0825_1542}.

\begin{corollary*}[$U$ concave]
    The continuation kernel $U$ is strictly concave in $q$. Specifically, it admits an unique maximizer $q^* = \text{argmax}_q U(F, M, q)$.
\end{corollary*}

\begin{proof}
    The continuation kernel is defined as $U(F,M,q) = -c(q) + \beta \, V\big( T_F(F, q), T_M(F, M) \big)$. By Assumption \ref{ass:cost}, $-c(q)$ is strictly concave. Moreover, $T_F(F, q) = (1-\lambda) \, F + \lambda \, s(q)$ is concave by Assumption \ref{ass:update}. Thus, $V\big( T_F(F, q), T_M(F, M) \big)$ is concave as the composition of concave and increasing functions. Overall, $U$ is strictly concave. This implies that its maximizer is unique.
\end{proof}

Finally, we discuss the proof of Proposition \ref{prop:0512_2148}.

\begin{proposition*}[Optimal Effort - Comparative Statics]
Assume \ref{ass:update}--\ref{ass:cost}. Additionally, assume \ref{ass:honest}. Then, the optimal effort adopted by the agent, $q^*(F,M; \kappa^S) = argmax_{q \in [0,1]} U(F, M, q)$, is such that:
\begin{itemize}
    \item $q^*(F,M; \kappa^S)$ is non-increasing in $F$.
    \item $q^*(F,M; \kappa^S)$ is non-decreasing in $M$.
    \item $q^*(F,M; \kappa^S)$ is non-increasing in $\kappa^S$.
\end{itemize}
\end{proposition*}

\begin{proof}
    First, we observe that the monotonicity in the statement follows from the supermodularity or submodularity of the function $U$ -- with a slight abuse of notation, we consider $U$ as a four-variate function: $U=U(F,M, \kappa^S, q)$. Then, we observe that, due to the fact that $U(F,M, \kappa^S, q) = -c(q) + \beta \, V\big( T(F,M, q), \kappa^S \big)$ and $c$ only depends on $q$, the super/submodularity of $U$ can be obtained by looking at the super/submodularity of $V$. In the following, we analyze the three points separately.\\
    \linebreak
    $\bm{\triangleright}$ We show that $q^*$ is non-increasing in $F$ by leveraging the Topkis's theorem. Thus, we have to show that $U(F,M,q)$ is submodular in $(F,q)$. To prove this, we can consider $M$ as a constant. So, for this piece of proof, with a slight abuse of notation, we neglect the explicit dependence of $U$, $W$, and $T$ on $M$ -- thus, we write $U(F,q)$, $W(F, q)$, and $T(F,q)$. We observe that the submodularity condition can be written as:
    \begin{equation*}
        \mathcal{M}(U, F_{1,2}, q_{1,2}) / \beta = V\big(T(F_1, q_1) \big) + V\big(T(F_2, q_2) \big) - V\big(T(F_1, q_2) \big) - V\big(T(F_2, q_1) \big) \le 0
    \end{equation*}
    where we only care about the $F$ component of $T$: $T_F(F_j, q_i) = (1 - \lambda) \, F_j + \lambda \, s(q_i)$. Proposition \ref{prop:0509_0331} discusses the concavity of the value function $V$. Thus, we can exploit a well-known property of concave function: given $V(F)$ concave, $\tilde{F}_1 \le \tilde{F}_2$, and $\tilde{F}_\Delta >0$, then $V(\tilde{F}_2 + \tilde{F}_\Delta) - V(\tilde{F}_2) \le V(\tilde{F}_1 + \tilde{F}_\Delta) - V(\tilde{F}_1)$. This is enough for the submodularity. Indeed, we consider $\tilde{F}_j = (1-\lambda) F_1 + \lambda \, s(q_j)$, with $j=1,2$, and $\tilde{F}_\Delta = (1 - \lambda) \, (F_2 - F_1)$. Thus, we obtain that $\mathcal{M}(U, F_{1,2}, q_{1,2}) \le 0$. This allows us to say that $U$ is submodular in $(F,q)$, so the optimal effort $q^* = argmax_q U(F,M,q)$ is non-increasing in $F$.\\
    \linebreak
    $\bm{\triangleright}$ To show that $q^*$ is non-decreasing in $M$ we have to prove that $U(M,q)$ is supermodular. Furthermore, $U$ is supermodular in $(M,q)$ if $V$ is supermodular in $(F,M)$. Indeed, $T_F$ (monotonically increasing) depends on $q$ and is independent on $M$; $T_M$ only depends on $M$. To prove the supermodularity of $V$, we exploit a Banach argumentation. Let's consider the set
    \begin{equation*}
        S_M = \big\{ W \in C_b(\mathcal{X}) \ | \ \forall F_1 \le F_2, \ \forall M_1 \le M_2, \ \mathcal{M}(W, F_{1,2}, M_{1,2}) \ge 0\big\}
    \end{equation*}
    The set is trivially non-empty, as the constant function is in it. To prove that $S$ is closed, we consider a sequence, $\{W_n\}$, of elements of $S_M$ converging to a certain $W$. Then, $\forall \varepsilon>0$, for sufficiently large $n$ we have that $W_n(F,M) - \varepsilon < W(F,M)$ and $ W(F,M) - \varepsilon < W_n(F,M)$, so
    \begin{align*}
        \mathcal{M}(W, F_{1,2}, M_{1,2}) = & W(F_1,M_1) + W(F_2,M_2) - W(F_1,M_2) - W(F_2,M_1) \ge \\
        & W_n(F_1,M_1) + W_n(F_2,M_2) - W_n(F_1,M_2) - W_n(F_2,M_1) - 4 \, \varepsilon \ge - 4 \, \varepsilon
    \end{align*}
    and letting $\varepsilon \rightarrow 0$ we have $W \in S_M$. Lastly, we have to show the invariance. We rewrite the supermodularity condition for $\mathcal{B}^{cont}W$ as:
    \begin{equation*}
        \mathcal{M}(\mathcal{B}^{cont}W, F_{1,2}, M_{1,2}) = (f-a \, \kappa^S) \, \mathcal{M}(v, F_{1,2}, M_{1,2}) + \mathcal{M}\bigg( \max_{q \in [0,1]} U_W(F, M, q), F_{1,2}, M_{1,2} \bigg)
    \end{equation*}
    The first term in the sum is greater than $0$ as the volume function $v$ is supermodular, as $\partial_{F, M} v \ge 0$ by Assumption \ref{ass:volume}. As for the second term, we have to show the supermodularity of the function $\max U_W$. We define $z=T_F(F,q)=(1-\lambda) \, F + \lambda \, s(q)$ and substitute $q$ with $s^{-1} \big( (z - (1-\lambda) \, F) / \lambda \big)$. The existence of $s^{-1}$ (increasing) is guaranteed by the fact that $s(q)$ is strictly increasing, as in Assumption \ref{ass:update}. Moreover, $s(q)$ is concave, so $s^{-1}$ is convex. Keeping in mind that $z \in \big[ z_{min}=T_F(F,0), \ z_{max}=T_F(F,1) \big]$, we can write:
    \begin{equation*}
        \max_{q \in [0,1]} U_W(F, M, q) = \max_{z \in [z_{min}, z_{max}]} \bigg[ -c\big( s^{-1}\big( u(F, z) / \lambda \big) \big) + \beta \, W\big( z, T_M(M) \big) \bigg]
    \end{equation*}
    where $u(F,z) = z - (1-\lambda) \, F$. To prove that the maximum over $z$ is supermodular in $(F,M)$, we exploit a result in \cite{topkis1998supermodularity}\footnote{Specifically, Theorem 2.7.6. In our setting (we are working with real intervals), it holds as long as the $z$ domain is increasing in $(F,M)$, which is trivially verified as neither $z_{min}$ nor $z_{max}$ depends on $M$, and they are $T_F(F,M)$ is increasing in $F$, $\forall M \in [0,1]$.}, guaranteeing that the $\max$ is supermodular in $(F,M)$ if its argument is supermodular pair-wise. As for $(F,M)$, we just observe that their contribution is disjoint. Regarding $(M, z)$, the supermodularity is given by the invariance hypothesis. Lastly, we consider $(F,z)$. Due to the change of variables we made, $W$ does not depend on $F$, thus we only have to focus on showing that $\mathcal{M}\bigg( -c(s^{-1}(u)), F_{1,2}, z_{1,2} \bigg) \ge 0$ or, equivalently, $\mathcal{M}\bigg( c(s^{-1}(u)), F_{1,2}, z_{1,2} \bigg) \le 0$. The function $c(s^{-1})$ is convex (composition of convex and strictly increasing functions), so we aim to apply the symmetrical result as the previous point (that is, the fixed-length increments of a convex function are increasing) with $\tilde{u}_j = z_j - (1-\lambda) \, F_2$ with $j=1,2$ and $\tilde{u}_\Delta = (1 - \lambda) \, (F_2 - F_1)$. Thus, we obtain:
    \begin{equation*}
        c\big(s^{-1}(\tilde{u}_2 + \Delta) \big) - c\big(s^{-1}(\tilde{u}_2) \big) \ge c\big(s^{-1}(\tilde{u}_1 + \Delta) \big) - c\big(s^{-1}(\tilde{u}_1) \big)
    \end{equation*}
    That is:
    \begin{equation*}
        c\big(s^{-1}(u(z_2, F_1)) \big) - c\big(s^{-1}(u(z_2, F_2)) \big) \ge c\big(s^{-1}(u(z_1, F_1)) \big) - c\big(s^{-1}(u(z_1, F_2)) \big)
    \end{equation*}
    which is exactly $\mathcal{M}\bigg( c(s^{-1}(u)), F_{1,2}, z_{1,2} \bigg) \le 0$. Ultimately, we have the pairwise supermodularity of $U$, which leads to the supermodularity of $\max_q U$, which implies $\mathcal{B}^{cont}W \in S_M$. Thus, $S_M$ is non-empty, closed, and invariant under $\mathcal{B}^{cont}$. By the Banach theorem, the fixed point of $\mathcal{B}^{cont}$, that is $V^{cont} = V$, is in $S_M$, thus $V$ is supermodular. Ultimately, $U(F, M, q)$ is supermodular in $(M,q)$. This implies that the optimal $q^* = argmax_q U(F, M, q)$ is non-decreasing in $M$. \\
    \linebreak
    $\bm{\triangleright}$ Finally, we investigate the monotonicity of $q^*$ as a function of $\kappa^S$. Without re-explicating all the computations, we extend the definition of the value function to consider $\kappa^S$ as an additional input. With a slight abuse of notation, we omit to explicitly display the dependence of $V$ and $T$ on the state variables $F$ and $M$. The target is showing that $V$ is in
    \begin{equation*}
        S_{\kappa^S} = \big\{ W \in C_b(\mathcal{X} \times [0,1]) \ | \ \mathbf{(1)} \ W \text{ concave in $F$ and } \ \mathbf{(2)} \ \forall \kappa^S_1 \le \kappa^S_2, \ \forall q_1 \le q_2, \ \mathcal{M}(U_W, \kappa^S_{1,2}, q_{1,2}) \le 0\big\}
    \end{equation*}
    Concavity in $F$ at every fixed $(M, \kappa^S)$ is preserved by $\mathcal{B}^{cont}$ on $C_b\big( \mathcal{X} \times [0,1] \big)$. Thus, we focus on the submodularity. Replicating the same reasoning as above, we can show that $S_{\kappa^S}$ is non-empty and closed. It remains to show the invariance for $\mathcal{B}^{cont}$. We can decompose:
    \begin{align*}
        \mathcal{M}\big( U_{\mathcal{B}^{cont} W}, \kappa^S_{1,2}, q_{1,2} \big) =& a \, \beta \, \mathcal{M}\big( - \kappa^S \, v\big( T(F, M, q) \big), \kappa^S_{1,2}, q_{1,2} \big) +\\
        & + \beta \, \mathcal{M} \bigg( \max_{u \in [0,1]} \bigg[ -c(u) + \beta \, W\big( T\big( T(F, M, q), u \big) \big) \bigg], \kappa^S_{1,2}, q_{1,2} \bigg)
    \end{align*}
    As for the first term of the sum, we observe that, thanks to the smoothness assumptions on $v$ and $s$ in Assumptions \ref{ass:volume} and \ref{ass:update}, $- \kappa^S \, v\big( T(F, M, q) \big)$ is $C^2$ in $(\kappa^S, q)$, and its mixed second derivative is $- \lambda \, s'(q) \, \partial_F v < 0$, so the left term is negative\footnote{A well-known sufficient condition for submodularity of a $C^2$ function states that if a function $f(x,y)$ is twice derivable and $\partial^2_{x,y} f \le 0$, then $f$ is submodular.}. As for the second term of the sum, we rewrite
    \begin{equation*}
        T\big( T(F, M, q), u \big) = (1-\lambda)^2 \, F + (1-\lambda) \, \lambda \, s(q) + \lambda \, s(u)
    \end{equation*}
    Then, we observe that the submodularity of $\max_u$ is obtained by showing the pairwise submodularity of the $\max$ argument, thanks to the same theorem applied in the "$q^*$ non-increasing in $M$" case. The pairs $(q, \kappa^S)$ and $(u, \kappa^S)$ lead to submodularity by the invariance hypothesis $W \in S_{\kappa^S}$. As for the pair $(q,u)$, we exploit the same result on the increments of concave functions. Specifically, we have the concavity of $W$ with respect to $F$, and we consider $\tilde{F}_j = (1-\lambda)^2 F + (1-\lambda) \, \lambda \, s(q_j) + \lambda \, s(u_1)$ and $\tilde{F}_\Delta = \lambda \, [s(u_2) - s(u_1)]$. Ultimately, we have shown $\mathcal{M}\big( U_{\mathcal{B}^{cont} W}, \kappa^S_{1,2}, q_{1,2} \big) \le 0$. As a consequence, $U$ is submodular in $(q, \kappa^S)$ and the optimal effort $q^*$ is non-increasing in the stake ratio $\kappa^S$.
\end{proof}

\end{document}